\documentclass[12pt]{article}
\usepackage{palatino}
\usepackage{StyleSheets/template}
\usepackage{StyleSheets/moreStyle}
\usepackage{StyleSheets/global}
\usepackage{subfiles}
\usepackage{hyperref}
\usepackage{url,color}
\usepackage{footnote}
\usepackage{threeparttable}
\makesavenoteenv{tabular}
\makesavenoteenv{table}

\begin{document}

\title{Indistinguishability Obfuscation from Well-Founded Assumptions} 

\author{
    Aayush Jain\thanks{UCLA, Center for Encrypted Functionalities, and NTT Research. Email: \texttt{aayushjain@cs.ucla.edu}. }\\
\and
    Huijia Lin\thanks{UW. Email: \texttt{rachel@cs.washington.edu}. }\\
\and
    Amit Sahai\thanks{UCLA, Center for Encrypted Functionalities. Email: \texttt{sahai@cs.ucla.edu}.}\\
}
\date{February, 2020} 
\date{August 18, 2020}
\maketitle
 
\noteswarning 

\thispagestyle{empty}
\newtheorem*{theorem*}{Theorem}
\newcommand{\lwe}{\epsilon}

\begin{abstract}
  In this work, we show how to construct indistinguishability
  obfuscation from subexponential hardness of four well-founded
  assumptions. We prove:

\begin{theorem*}[Informal]
    Let $\tau \in (0,\infty), \delta \in (0,1), \lwe \in (0,1)$ be arbitrary constants. Assume sub-exponential security of
    the following assumptions, where $\secparam$ is a security parameter, and the parameters $\ell,k,n$ below are large enough polynomials in $\lambda$: 
  \begin{itemize}
  \item the $\mathsf{SXDH}$ assumption on asymmetric bilinear groups of a prime order $p = O(2^\secparam)$,
  \item the $\mathsf{LWE}$ assumption over $\Int_p$ with subexponential modulus-to-noise ratio $2^{k^\lwe}$, where $k$ is the dimension of the $\mathsf{LWE}$ secret,
  \item the $\mathsf{LPN}$ assumption over $\mathbb{Z}_p$ with
    polynomially many $\mathsf{LPN}$ samples and error rate
    $1/\ell^\delta$, where
    $\ell$ is the dimension of the $\mathsf{LPN}$ secret,
  \item the existence of a Boolean $\mathsf{PRG}$ in $\mathsf{NC}^0$ with
    stretch $n^{1+\tau}$,
  \end{itemize}
  Then, (subexponentially secure) indistinguishability obfuscation for
  all polynomial-size circuits exists.
  \end{theorem*}
  
  Further, assuming only polynomial security of the aforementioned
  assumptions, there exists collusion resistant public-key functional encryption
  for all polynomial-size circuits. 
\end{abstract}

\pagebreak
\tableofcontents
\thispagestyle{empty}
\pagebreak
\setcounter{page}{1}

\newcommand{\lpnstretch}{\epsilon}
\newcommand{\spasity}{\delta}
\newcommand{\flpn}{\mathsf{LPN}}
\newcommand{\pid}{I}
\newcommand{\pbid}{{\bar I}}
\newcommand{\piset}{\mathcal{I}}
\newcommand{\pg}{G}
\newcommand{\pset}{\mathcal{G}}
\newcommand{\psamp}{\mathsf{Samp}}
\newcommand{\pidsamp}{\mathsf{IdSamp}}
\newcommand{\psdsamp}{\mathsf{SdSamp}}
\newcommand{\peval}{\mathsf{Eval}}
\newcommand{\pprocess}{\mathsf{Process}}
\newcommand{\pstretch}{\tau}
\newcommand{\pdeg}{d}
\newcommand{\plocality}{c}
\newcommand{\sparsity}{\delta}
\newcommand{\modulus}{\epsilon_{p}}
\newcommand{\explain}[1]{{\em Note: #1}}
\newcommand{\ol}[1]{{\overline{#1}}}
\newcommand{\psd}{P}
\newcommand{\ssd}{S}
\newcommand{\seclevel}{\gamma}
\newcommand{\Nat}{\mathbb{N}}
\newcommand{\pprg}{\mathsf{pPRG}}
\newcommand{\sprg}{\mathsf{sPRG}}
\section{Introduction}

In this work, we study the notion of indistinguishability obfuscation
($\iO$) for general polynomial-size circuits
\cite{C:BGIRSVY01,C:GolKalRot08,FOCS:GGHRSW13}.  $\iO$ requires that for any two 
circuits $\mathsf{C}_0$ and $\mathsf{C}_1$  of the
same size, such that
$\mathsf{C}_0(x) = \mathsf{C}_1(x)$ for all inputs $x$, we have that
$\iO(\mathsf{C}_0)$ is computationally indistinguishable to
$\iO(\mathsf{C}_1)$.  Furthermore, the obfuscator $\iO$ should be computable in
probabilistic polynomial time.
The notion of $\iO$ has proven to be very powerful, with over a hundred  papers published utilizing $\iO$ to enable a remarkable variety of applications in cryptography and complexity theory; indeed $\iO$ has even expanded the scope of cryptography,  (see, e.g. ~\cite{FOCS:GGHRSW13,SW14,C:BrzFarMit14,EC:GGGJKL14,C:HohSahWat13,KLW15,FOCS:BitPanRos15,CHNVW16,C:GarPanSri16,AC:HJKSWZ16}).

Despite this success, until this work, all previously known $\iO$ constructions~\cite{EC:GarGenHal13,FOCS:GGHRSW13, EC:BGKPS14,BR14,PST14,AGISfull,BMSZ16,C:CorLepTib13,C:CorLepTib15,TCC:GenGorHal15,CheonHaLeRySt14,BWZ14,CoronGHLMMRST15,HuJ15,BrakerskiGHLST15,Halevi15,CheonLR15,MinaudF15,MSZ16,DGO+16,EC:Lin16,FOCS:LinVai16,EC:AnaSah17,C:Lin17,C:LinTes17,GJK18,AJS18,EC:Agrawal19,LM18,EC:JLMS19,ITCS:BIJMSZ20,EC:AgrPel20,BDGM20} 
required new hardness assumptions that were postulated specifically for showing security of the $\iO$ schemes proposed.  
Indeed, the process of understanding these assumptions has been tortuous, with several of these assumptions  broken by clever cryptanalysis~\cite{CheonHaLeRySt14,BWZ14,CoronGHLMMRST15,HuJ15,BrakerskiGHLST15,Halevi15,CheonLR15,MinaudF15,MSZ16,BBKK17,TCC:LomVai17,EC:BHJKS19}. The remaining standing ones are based on new and novel computational problems that are different in nature from well-studied computational problems (for instance, $\mathsf{LWE}$ with leakage on noises).


As a result, there has been a lack of clarity about the state of $\iO$ security~\cite{BKMPRS19}. Our work aims to place $\iO$ on \emph{terra firma}. 

\paragraph{Our contribution.}
We show how to construct $\iO$ from subexponential hardness of four well-founded assumptions. We prove:
\begin{theorem} (Informal)
   Let $\tau$ be arbitrary constants greater than 0, and
    $\delta$, $\lwe$ in $(0,1)$. Assume sub-exponential security of
    the following assumptions, where $\secparam$ is the security parameter, and the parameters $\ell,k,n$ below are large enough polynomials in $\lambda$: 
  \begin{itemize}
  \item the $\mathsf{SXDH}$ assumption on asymmetric bilinear groups of a prime order $p = O(2^\secparam)$,
  \item the $\mathsf{LWE}$ assumption over $\Int_p$ with subexponential modulus-to-noise ratio $2^{k^\lwe}$, where $k$ is the dimension of the $\mathsf{LWE}$ secret,
  \item the $\mathsf{LPN}$ assumption over $\mathbb{Z}_p$ with
    polynomially many $\mathsf{LPN}$ samples and error rate
    $1/\ell^\delta$, where
    $\ell$ is the dimension of the $\mathsf{LPN}$ secret,
  \item the existence of a Boolean $\mathsf{PRG}$ in $\mathsf{NC}^0$ with
    stretch $n^{1+\tau}$,
  \end{itemize}

  Then, (subexponentially secure) indistinguishability obfuscation for
  all polynomial-size circuits exists.
  \end{theorem}
All four assumptions are based on computational problems with a long history of study, rooted in complexity, coding, and number theory. Further, they were introduced for building basic cryptographic primitives (such as public key encryption), and have  been used for realizing a variety of  cryptographic goals  that have nothing to do with $\iO$.   

\subsection{Assumptions in More Detail}
We now describe each of these assumptions in more detail and briefly survey their history.

\paragraph{The $\mathsf{SXDH}$ Assumption:} The standard $\mathsf{SXDH}$ assumption is stated as follows: Given an appropriate prime $p$, three groups $\mathbb{G}_1,\ \mathbb{G}_2$, and $\mathbb{G}_T$ 
are chosen of order $p$ such that there exists an efficiently computable nontrivial bilinear map $e: \mathbb{G}_1 \times \mathbb{G}_2 \rightarrow \mathbb{G}_T$. Canonical generators, $g_1$ for $\mathbb{G}_1$, and $g_2$ for 
$\mathbb{G}_1$, are also computed.
Then, the $\mathsf{SXDH}$ assumption requires that the Decisional Diffie Hellman (DDH) assumption holds in both $\mathbb{G}_1$ and $\mathbb{G}_2$. That is, it requires that
the following  computational indistinguishability holds:
\begin{align*}
\forall b \in \{1, 2\},\ \left \{\left (g_b^x, g_b^y, g_b^{xy}\right ) \ \mid \ x,y \gets \Z_p\right \}  \approx_c \left \{\left (g_b^x, g_b^y, g_b^{z} \right )\ \mid \ x,y,z \gets \Z_p
\right \}
\end{align*}

This assumption was first defined in the 2005 work of Ballard et. al. \cite{BGMM05}. Since then, $\mathsf{SXDH}$ has seen extensive use in a wide variety of applications throughout cryptography, including Identity-Based Encryption and Non-Interactive Zero Knowledge (See, e.g. \cite{EC:GroSah08,FOCS:BKKV10,AC:BisJaiKow15,C:Lin17,CLLWW12,AC:JutRoy13}).
It has been a subject of extensive cryptanalytic study (see~\cite{EC:Verheul01} for early work and~\cite{GR04} for a survey).

\paragraph{The $\mathsf{LWE}$ Assumption:} The $\mathsf{LWE}$ assumption with respect to subexponential-size modulus $p$, dimension $\secparam$, sample complexity $n(\secparam)$ and polynomial-expectation discrete Gaussian distribution $\chi$ over integers states that the following computational indistinguishability holds:
\begin{align*}
    \{\mat A, \vec s \cdot \mat{A}+\vec{e} \mod p\ \mid\ & \mat{A}\leftarrow \Z^{\secparam\times n}_{p},\ \vec s \leftarrow \Z^{1\times \secparam}_p,\ \vec e \leftarrow \chi^{1\times n} \}\\
    \approx_c\ \{\mat A,\vec u\ \mid\ & \mat{A}\leftarrow \Z^{\secparam \times n}_{p},\ \vec u \leftarrow \Z_p^{1\times n}\}
\end{align*}

This assumption was first stated in the work of \cite{Regev05}. The version stated above is provably hard as long as $\mathsf{GAP}$-$\mathsf{SVP}$. 
 is hard to approximate to within subexponential factors in the worst case
 \cite{Regev05,STOC:Pei09,STOC:GenPeiVai08,FOCS:MicReg04,C:MicPei13}. $\mathsf{LWE}$ has been used extensively to construct applications such as Leveled Fully Homomorphic Encryption \cite{FOCS:BraVai11,ITCS:BraGenVai12,C:GenSahWat13}, Key-Homomorphic PRFs \cite{C:BLMR13}, Lockable Obfuscation \cite{FOCS:GoyKopWat17,FOCS:WicZir17}, Homomorphic Secret-Sharing \cite{EC:MukWic16,C:DHRW16}, Constrained PRFs \cite{TCC:BraVai15}, Attribute Based Encryption \cite{BGG+14,STOC:GorVaiWee13,C:GorVaiWee15} and Universal Thresholdizers \cite{C:BGGJ18}, to name a few.

\paragraph{The existence of $\mathsf{PRG}$s in $\mathsf{NC}^0$:} The assumption of the existence of a Boolean $\mathsf{PRG}$
in $\mathsf{NC}^0$
states that there exists a Boolean function ${G}:\{0,1\}^{n}\rightarrow \{0,1\}^m$ where $m=n^{1+\pstretch}$ for some constant $\pstretch>0$, and where each output bit computed by $\mathsf{G}$ depends on a constant number of input bits, such that the following computational indistinguishability holds:
\begin{align*}
 \{ G(\vec \sigma)\ \mid \ \vec \sigma \leftarrow {\left \{0,1\right \}}^n \} 
\approx_c
 \{ \vec y\ \mid\ \vec y \leftarrow {\left \{0,1\right \}}^m \} 
\end{align*}
Pseudorandom generators are a fundamental primitive in their own right, and have vast applications throughout cryptography.
$\mathsf{PRG}$s in $\mathsf{NC}^0$ are tightly connected to the
fundamental topic of Constraint Satisfaction Problems (CSPs) in
complexity theory, and were first proposed for cryptographic use by
Goldreich~\cite{Gol00,CM01} 20 years ago. The complexity theory and
cryptography communities have jointly developed a rich body of
literature on the cryptanalysis and theory of constant-locality
Boolean PRGs~\cite{Gol00,CM01,FOCS:MosShpTre03,TCC:AppBogRos12,BogdanovQ12,STOC:Applebaum12,CCC:OdoWit14,STOC:AppLov16,STOC:KMOW17,AC:CDMRR18}.

\paragraph{$\mathsf{LPN}$ over large fields:} Like $\mathsf{LWE}$, the $\mathsf{LPN}$ assumption over finite fields $\Z_p$ is also a decoding problem. The standard $\mathsf{LPN}$ assumption with respect to subexponential-size modulus $p$, dimension $\ell$, sample complexity $n(\ell)$ and a noise rate $r=1/\ell^{\delta}$ for $\delta \in (0,1)$  states that the following computational indistinguishability holds:
\begin{align*}
    \{\mat A, \vec s \cdot \mat{A}+\vec{e} \mod p\ \mid\ & \mat{A}\leftarrow \Z^{\ell\times n}_{p},\ \vec s \leftarrow \Z^{1\times \ell}_p,\ \vec e \leftarrow \D_{r}^{1\times n}\}\\
    \approx_c\
     \{\mat A,\vec u\ \mid \ & \mat{A}\leftarrow \Z^{\ell\times n}_{p},\ \vec u \leftarrow \Z_p^{1\times n}\}.
\end{align*}
Above $e\gets \D_{r}$ is a generalized Bernoulli distribution, \emph{i.e.} $e$ is sampled randomly from $\Z_p$ with probability ${1/\ell^{\delta}}$ and set to be $0$ otherwise. Thus, the difference between $\mathsf{LWE}$ and $\mathsf{LPN}$ is the structure of the error distribution. In $\mathsf{LWE}$ the error vector is a random (polynomially) bounded vector. In $\mathsf{LPN}$, it is a sparse random vector, but where it is nonzero, the entries have large expectation. 
The origins of the $\mathsf{LPN}$ assumption date all the way back to the 1950s: the works of Gilbert~\cite{Gilbert} and Varshamov~\cite{Varshamov} showed that random linear codes possessed remarkably strong minimum distance properties. However, since then, almost no progress has been made in efficiently decoding random linear codes under random errors. 
The $\mathsf{LPN}$ over fields assumption above formalizes this, and was formally defined for general parameters in 2009~\cite{IPS09}, under the name ``Assumption 2.'' While in~\cite{IPS09}, the assumption was used when the error rate is constant, in fact, polynomially low error (in fact $\delta = 1/2$) has an even longer history in the $\mathsf{LPN}$ literature: it was used by Alekhnovitch in 2003~\cite{FOCS:Alekhnovich03}  to construct public-key encryption with the field $\F_2$. The exact parameter settings that we describe above, with both general fields and polynomially low error, was explicitly posed by~\cite{CCS:BCGI18}. 

This assumption was posed for the purpose of building efficient secure two-party and multi-party protocols for arithmetic computations~\cite{IPS09,ITCS:AppAvrBrz15}. Earlier, $\mathsf{LPN}$ over binary fields was posed for the purpose of constructing identification schemes~\cite{AC:HopBlu01} and public-key encryption~\cite{FOCS:Alekhnovich03}. Recently, the assumption has led to a wide variety of applications (see for example, \cite{IPS09,ITCS:AppAvrBrz15,CCS:BCGI18,C:ADINZ17,CCS:DGNNT17,AC:GhoNieNil17,EC:BLMZ19,CCS:BCGIKRS19}).
A comprehensive review of known attacks on $\mathsf{LPN}$ over large fields, for the parameter settings we are interested in, was given in~\cite{CCS:BCGI18}. For our parameter setting, the best running time of known attacks is sub-exponential, for any choice of the constant $\delta \in (0,1)$ and for any polynomial $n(\ell)$

\subsection{Our Ideas in a Nutshell}
Previous work~\cite{AJS18,LM18,C:AJLMS19,EC:JLMS19,JLS19,GJLS20}
showed that to achieve $\iO$, it is sufficient to assume $\mathsf{LWE}$, $\mathsf{SXDH}$,
and
$\mathsf{PRG}$
in
$\mathsf{NC}^0$,
and one other object, that we will encapsulate as a \emph{structured-seed} PRG ($\mathsf{sPRG}$) with polynomial stretch and special efficiency properties. In an $\mathsf{sPRG}$, 
the seed to the 
$\mathsf{sPRG}$ consists of both a public and private part. The pseudorandomness property of the $\mathsf{sPRG}$ should hold even when the adversary can see the public seed in addition to the output of the $\mathsf{sPRG}$. Crucially, the output of the $\mathsf{sPRG}$
should be computable by a \emph{degree-2} computation in the private seed (where, say, the coefficients of this degree-2 computation are obtained through constant-degree computations on the public seed).

Our key innovation is a simple way to leverage $\mathsf{LPN}$ over fields to build an $\mathsf{sPRG}$. The starting point for our construction is the following observation.
Assuming $\mathsf{LPN}$ and that $G$ is an (ordinary) $\mathsf{PRG}$ in $\mathsf{NC}^0$ with  stretch $m(n)$, we immediately have the following computational indistinguishability:
\begin{align*}
 \Big\{(\vec{A}, \; \vec{b}= \vec{s}\cdot \vec{A}+ \vec{e} +
            \vec{\sigma},\; G(\vec{\sigma}))  \ \mid \ & \vec{A}\leftarrow \Z^{\ell\times n}_{p};\;
\vec{s}\leftarrow \Z^{1\times \ell}_{p};\;
\vec{e} \leftarrow \cD^{1 \times n}_{r}(p);\;
\vec{\sigma}\leftarrow \{0,1\}^{1 \times n}\Big \}\\
 \approx_c\ \Big \{ (\vec{A},\; \vec{u},\; \vec{w})\ \mid \ & \vec{A}\leftarrow \Z^{\ell\times n}_{p};\;
\vec{u}\leftarrow \Z^{1\times n}_{p};\; \vec{w} \leftarrow \{0,1\}^{1\times m(n)} \Big \}
\end{align*}

Roughly speaking, we can think of both $\vec{A}$ and $\vec{b}$ above as being public. All that remains is to show that the computation of $G(\vec{\sigma})$ can be performed using a degree-2 computation in a short-enough specially-prepared secret seed. Because $G$ is an arbitrary $\mathsf{PRG}$ in $\mathsf{NC}^0$, it will not in general be computable by a degree-2 polynomial in $\vec \sigma$. To accomplish this goal, we crucially leverage the \emph{sparseness} of the $\mathsf{LPN}$  error $\vec{e}$, by means of a simple pre-computation idea to ``correct'' for  errors introduced due to this sparse error. A gentle overview is provided in Section~\ref{sec:sprg}, followed by our detailed construction and analysis.

\newcommand{\blen}[1]{{\left\vert #1 \right\vert}}
\renewcommand{\dim}[1]{{\mathsf{dim}(#1)}}

\section{Preliminaries}
For any distribution $\mathcal{X}$, we denote by $x \leftarrow \mathcal{X}$  the process of sampling a value $x$ from the distribution $\mathcal{X}$. Similarly, for a set $X$ we denote by $x \leftarrow X$ the process of sampling  
$x$ from the uniform distribution over $X$. For an integer $n \in \mathbb{N}$ we denote by $[n]$ the set $\{1,..,n\}$. A function $\negl: \mathbb{N} \rightarrow \mathbb{R}$ is negligible if for every constant $c >0$ there exists an integer $N_c$ such that $\negl(\secparam)<\secparam^{-c}$ for all $\secparam>N_c$. Throughout, when we refer to polynomials in security parameter, we mean constant degree polynomials that take positive value on non negative inputs. We denote by $\poly(\secparam)$ an arbitrary polynomial in $\secparam$  satisfying the above requirements of non-negativity. We denote vectors by bold-faced letters such as $\vec{b}$ and $\vec{u}$. Matrices will be denoted by capitalized bold-faced letters for such as $\vec{A}$ and $\vec{M}$. 
For any $k\in \mathbb{N}$, we denote by the tensor product $\vec v^{\otimes k} = \underbrace{\vec v \otimes \dots \otimes
    \vecv}_{k}$ to be the standard tensor product, but converted back into a vector. We also introduce two new notations. First, for any vector $\vec{v}$ we refer by $\dim{\vec v}$ the dimension of vector $\vec v$. For any matrix $\mat{M} \in \Z^{n_1\times n_2}_q$, we denote by $\blen{\mat{M}}$ the  bit length of $\mat{M}$. In this case, $\blen{\mat{M}}=n_1\cdot n_2 \cdot \log_2 q$. We also overload this operator in that, for any set $S$, we use $\blen{S}$ to denote the cardinality of $S$. The meaning should be inferred from context.

For any two polynomials  $a(\secparam,n), b(\secparam,n):\mathbb{N} \times \mathbb{N}\rightarrow \mathbb{R}^{\geq 0}$, we say that $a$ is polynomially smaller than $b$, denoted as $a \ll b$, if there exists an $\epsilon \in (0,1)$ and a constant $c>0$ such that $a < b^{1-\epsilon} \cdot \lambda^c$ for all large enough $n,\lambda \in \mathbb{N}$.
The intuition behind this definition is to think of $n$ as being a sufficiently large polynomial in $\lambda$

\paragraph{Multilinear Representation of Polynomials and Representation over $\Z_p$.} In this work we will consider multivariate polynomials $p \in \Z[\vec x= (x_1,\ldots,x_n)]$ mapping $\{0,1\}^n$ to $\{0,1\}$. For any such polynomial there is a unique multilinear polynomial $p'$ (obtained by setting $x^2_i=x_i$) such that $p'\in \Z[\vec x]$ and $p'(\vec x)=p(\vec x)$ for all $\vec x\in \{0,1\}^n$. Further, such a polynomial can have a maximum degree of $n$. At times, we will consider polynomials $g\in \Z_p[\vec x]$ such that for every $\vec x \in \{0,1\}^n$, $g(\vec x) \mod p= p(\vec x)$. Such a polynomial $g$ can be constructed simply as follows. Let $p'(\vec x)= \sum_{S\subseteq [n]} c_{S} \ \Pi_{i\in S} x_i $. We can construct $g(\vec x)=\sum_{S\subseteq [n]} (c_{S} \mod p)\Pi_{i\in S} x_i$. Note that $g$ has degree at most the degree of $p'$ over $\Z$.
For polynomials of degree $d$, both the process described above can take $O(n^{d})$ time. 
In this work, we consider polynomials representing pseudorandom generators in $\mathsf{NC}^0$. Such polynomials depend only on a constant number of input bits, and thus their multilinear representations (and their field representations) are also constant degree polynomials. In this scenario, these conversions take polynomial time.

  \begin{definition}[$(T,\epsilon)$-indistinguishability]
  We say that two ensembles $\mathcal{X}=\{\mathcal{X}_{\secparam}\}_{\secparam \in \mathbb{N}}$ and $\mathcal{Y}=\{\mathcal{Y}_{\secparam}\}_{\secparam \in \mathbb{N}}$ are $(T,\epsilon)$-indistinguishable where $T: \mathbb{N}\rightarrow \mathbb{N}$ and $\epsilon: \mathbb{N}\rightarrow [0,1]$ if for every non-negative polynomial $\poly(\cdot,\cdot)$ and any adversary $\adversary$  running in time bounded by $T\poly(\secparam)$ it holds that: For every sufficiently large $\secparam \in \mathbb{N}$,
\begin{align*}
\bigg | \Pr_{x \leftarrow  \mathcal{X}_{\secparam}  }[     \adversary(1^{\secparam},x)=1  ] -\Pr_{y \leftarrow  \mathcal{Y}_{\secparam}} [\adversary(1^{\secparam},y)=1] \bigg | \leq \epsilon(\secparam).
\end{align*}

We say that two ensembles are $\epsilon$-indistinguishable if it is
$(\secparam,\epsilon)$-indistinguishable,
and is subexponentially $\epsilon$-indistinguishable if it is
$(T,\epsilon)$-indistinguishable for $T(\secparam)=2^{\secparam^c}$ for some positive constant
$c$. It is indistinguishable if it is
$\frac{1}{\secparam^c}$-pseudorandom for every
positive constant $c$, and subexponentially indistinguishable if $(T, 1/T)$-indistinguishable for $T(\secparam)=2^{\secparam^c}$ for some positive constant $c$.
\end{definition}
Below if the security a primitive or the hardness of an assumption are defined through indistinguishability, we say the primitive or assumption is $(T,\epsilon)$ secure, hard, or indistinguishable, or (subexponentially) secure, hard, or indistinguishable if the appropriate $(T,\epsilon)$-indistinguishability or (subexponentially) indistinguishability holds.


\paragraph{Indistinguishability Obfuscation.}
We now define our object of interest, Indistinguishability Obfuscation ($\iO$).
The notion of indistinguishability obfuscation (iO), first conceived by Barak et al.~\cite{BGIRSVY01}, guarantees that the obfuscation of two circuits are computationally indistinguishable as long as they both are equivalent circuits, i.e., the output of both the circuits are the same on every input. Formally, 

\newcommand{\inpl}{n}
\begin{definition} [Indistinguishability Obfuscator (iO) for Circuits]
\label{def:io} A uniform PPT algorithm $\iO$ is called a $(T,\gamma)$-secure indistinguishability obfuscator for polynomial-sized circuits 
if the following holds:
\begin{itemize}

\item {\bf Completeness:} For every $\secparam \in \mathbb{N}$, every circuit $C$ with input length $n$, every input $x \in \{0,1\}^{\inpl}$, we have that 
$$\prob \left[C'(x) = C(x)\ :\ C' \leftarrow \iO(1^\secparam,C) \right] = 1~.$$

\item {$(T,\gamma)$-{\bf Indistinguishability:}}  For every two ensembles  $\{C_{0,\secparam}\}$ $\{C_{1,\secparam}\}$ of polynomial-sized circuits that have the same size, input length, and output length, and are functionally equivalent, that is, $\forall \secparam$, $C_{0,\secparam}(x) = C_{1,\secparam}(x)$ for every input $x$, the following distributions are  $(T,\gamma)$-indistinguishable. 
\begin{align*}
\{\iO(1^\secparam,C_{0,\secparam})\} \qquad \{\iO(1^\secparam,C_{1,\secparam})\}    
\end{align*}


\end{itemize}
\end{definition}

\paragraph{LPN over Fields Assumption.}
In this work, we use the LPN assumption over a large field. This assumption has been used in a various works (see for example, \cite{IPS09,ITCS:AppAvrBrz15,CCS:BCGI18,C:ADINZ17,CCS:DGNNT17,AC:GhoNieNil17,EC:BLMZ19,CCS:BCGIKRS19}). We adopt the following definition from \cite{CCS:BCGI18}.

We set up some notation for the definition below. Let $p$ be any prime modulus. We define the distribution $\cD_{r}(p)$ as the distribution that outputs $0$ with probability $1-r$ and a random element from $\Z_p$ with the remaining probability.

\begin{definition}[$\flpn(\ell,n,r,p)$-Assumption, ~\cite{IPS09,ITCS:AppAvrBrz15,CCS:BCGI18}]
Let $\secparam$ be the security parameter.
For an efficiently computable prime modulus $p(\secparam)$, dimension $\ell(\secparam)$, sample complexity $n(\ell)$,  and noise rate $r(n)$ we say that the $\flpn(\ell,n,r,p)$ assumption is $(T,\gamma)$-secure / hard / indistinguishable if the following two distributions are $(T,\gamma)$-indistinguishable:
\begin{align*}
 \Big \{ \left (\vec{A}, \vec{b}= \vec{s}\cdot \vec{A}+ \vec{e} \right )\ | & \ \vec{A}\leftarrow \Z^{\ell\times n}_{p}, \ \vec{s}\leftarrow \Z^{1\times \ell}_{p}, \
\vec{e} \leftarrow \cD^{1 \times n}_{r}(p) \Big \}\\
 \Big \{ \left (\vec{A}, \vec u \right )\ | &\ \vec{A}\leftarrow \Z^{\ell\times n}_{p}, \ \vec u \leftarrow \Z^{1\times n}_p\Big \}
\end{align*}
\end{definition}
 We will set $\ell$ to be a large enough polynomial in $\secparam$,  set $r = \ell^{-\delta}$, for a constant $\delta \in (0,1)$, and set the number of samples $n = \ell^{c}$ for some constant $c>1$. Note that this setting of parameters was considered in detail in the work of ~\cite{CCS:BCGI18}. 
We refer the reader to~\cite{CCS:BCGI18}  for a comprehensive discussion of the history and security of this assumption.

\paragraph{Leakage Lemma.}
We will use the following theorem in our security proofs.
   \begin{theorem}[Imported Theorem \cite{CCL18}] 
    \label{thm:leakage}
Let $n, \ell \in \bN, \epsilon > 0$, and $\cC_{leak}$ be a family of distinguisher circuits from $\{0,1\}^n \times \{0,1\}^\ell \rightarrow \{0,1\}$ of size $s(n)$. Then, for every distribution $(X, W)$ over $\{0,1\}^n \times \{0,1\}^\ell$, there exists a simulator $h$ such that: 
    \begin{enumerate}
        \item $h$ is computable by circuits of size bounded by $s' = O(s2^\ell \epsilon^{-2})$, and maps $\{0,1\}^n \times \zo^{s'}\rightarrow \{0,1\}^\ell$. We denote by $U$ the uniform distribution over $\zo^{s'}$.
        \item $(X,W)$ and $(X, h(X,U))$ are $\epsilon$-indistinguishable by $\cC_{leak}$. That is, for every $C \in \cC_{leak}$,
            \[\Abs{\Pr_{(x,w) \from (X,W)}[C(x, w) = 1] - \Pr_{x \from X, u \from U}[C(x, h(x,u)) = 1]} \leq \epsilon\]
    \end{enumerate}
\end{theorem}

\section{Definition of Structured-Seed PRG}

\begin{definition}[Syntax of Structured-Seed Pseudo-Random Generators (sPRG)]
 Let $\pstretch$ be a positive constant.  
A structured-seed Boolean PRG, $\sprg$, with stretch $\pstretch$ that maps $(n\cdot \poly(\secparam))$-bit binary strings into
$(m = n^\pstretch)$-bit strings, where $\poly$ is a fixed polynomial, is defined by the following
PPT algorithms:
\begin{itemize}
\item $\pidsamp(1^{\secparam}, 1^n)$ samples a function index $\pid$. 
\item $\psdsamp(\pid)$ jointly samples two binary strings, a public seed
  and a private seed, $\sd = (\psd,\ssd)$. The combined length of
  these strings is $n \cdot \mathsf{poly}(\lambda)$.
\item $\peval(\pid, \sd)$ computes a string in $\zo^m$.
\end{itemize}
\end{definition}
\begin{remark}[Polynomial Stretch.]
We denote an $\sprg$ to have polynomial stretch if $\pstretch>1$ for some constant $\pstretch$.
\end{remark}
\begin{remark}[On $\poly(\secparam)$ multiplicative factor in the seed length.]
As opposed to a standard Boolean $\prg$ definition where the length of the output is set to be $n^{\pstretch}$ where $n$ is the seed length, we allow the length of the seed to increase multiplicatively by a fixed polynomial $\poly$ in a parameter $\secparam$. Looking ahead, one should view $n$ as an arbitrary large polynomial in $\secparam$, and hence $\sprg$ will be expanding in length.
\end{remark}

\begin{definition}[Security of sPRG]
A structured-seed Boolean PRG, $\sprg$, satisfies 
\begin{description}
\item[$(T(\secparam),\seclevel (\secparam))$-pseudorandomness:] the following distributions are $(T,\seclevel)$ indistinguishable.
\begin{align*}
    \{ \pid,\ \psd,\ \peval(\pid,\psd)\ |& \  \pid \gets \pidsamp(1^{\secparam},1^n),\ \sd \gets \psdsamp(\pid)\}\\
    \{ \pid,\ \psd,\ \vec r\ |& \  \pid \gets \pidsamp(1^{\secparam},1^n),\ \sd \gets \psdsamp(\pid),\ \vec r \gets \{0,1\}^{m(n)}\}
\end{align*}
\end{description}
\end{definition}

\begin{definition}[Complexity and degree of $\sprg$]
 Let $\pdeg \in \Nat$, let $\secparam \in \mathbb{N}$ and $n=n(\secparam)$ be arbitrary positive polynomial in $\secparam$, and $p = p(\secparam)$ denote a prime modulus which is
  an efficiently computable function in $\secparam$.  Let $\mathbb{C}$ be a
  complexity class.  A $\sprg$ has complexity $\mathbb{C}$ in the
  public seed and degree $\pdeg$ in private seed over $\Int_p$,
  denoted as, $\sprg \in (\mathbb C, \text{ deg } \pdeg)$, if  for every $\pid$ in the support of
  $\pidsamp(1^{\secparam},1^n)$, there exists 
  an algorithm
  $\pprocess_\pid$ in $\mathbb{C}$  and an $m(n)$-tuple of  polynomials $Q_\pid$ that can be efficiently generated from $\pid$, such that for all $\sd$ in the support of $\psdsamp(\pid)$, it
  holds that:
  \begin{align*}
    \peval(\pid, \sd) = Q_\pid(\ol\psd, \ssd)\; \text{over } \Int_p\; , \; \ol\psd
    = \pprocess_\pid(\psd)\; ,
  \end{align*}
  where $Q_\pid$ has degree 1 in $\ol\psd$ and degree $\pdeg$ in $\ssd$. 
\end{definition}

We remark that the above definition generalizes the standard notion of
families of PRGs in two aspects: 1) the seed consists of a public part
and a private part, and 2) the seed may not be uniform. Therefore, we
obtain the standard notion as a special case. 
\begin{definition}[Pseudo-Random Generators, degree, and locality]
  A (uniform-seed) Boolean PRG ($\prg$) is an $\sprg$
  with a seed sampling algorithm $\psdsamp(\pid)$ that outputs
  a public seed $\psd$ that is an empty string and a uniformly random
  private seed $\ssd \gets \zo^n$, where the polynomial $\poly$ is fixed to be $1$.

  Let $\pdeg, \plocality \in \Nat$. The $\prg$ has multilinear degree $\pdeg$ if
  for every $\pid$ in the
  support of $\pidsamp(1^n)$, we have that $\peval(\pid,\sd)$ can be written as an $m(n)$-tuple of degree-$\pdeg$ polynomials over $\mathbb{Z}$ in $\ssd$. It has constant
  locality $\plocality$ if for every $n \in \Nat$ and $\pid$ in the
  support of $\pidsamp(1^n)$, every output bit of $\peval(\pid,\sd)$
  depends on at most $\plocality$ bits of $\ssd$.
\end{definition}


\newcommand{\lset}{\vars}    

\newcommand{\NC}{\mathsf{NC}}
\newcommand{\good}{\mathsf{good}}
\newcommand{\bad}{\mathsf{bad}}

\newcommand{\overflow}{{\epsilon_t}}

\newcommand{\gap}{\kappa}
\newcommand{\thresh}{T}
\newcommand{\cordeg}{t}
\newcommand{\Err}{\mathsf{ERR}}
\newcommand{\Cor}{\mathsf{BAD}}
\newcommand{\flag}{\mathsf{flag}}
\newcommand{\vars}{\mathsf{Vars}}
\newcommand{\Inp}{\mathsf{Inp}}
\renewcommand{\set}[1]{\left\{#1\right\}}
\newcommand{\mnl}{h}
\newcommand{\mnlu}{u}
\newcommand{\mnlv}{v}
\newcommand{\mnlset}{M}
\newcommand{\hyb}{\mathsf{H}}

\newcommand{\phibkt}{\phi_{\mathsf{bkt}}}
\newcommand{\phiind}{\phi_{\mathsf{ind}}}

\newcommand{\imbal}{\eta}
\newcommand{\uvec}{\vec{id}}
\newcommand{\arithNC}{{\mathsf{arith}\text{-}\NC}}

\renewcommand{\dim}[1]{{\mathsf{dim}(#1)}}  
\newcommand{\corr}{\mathsf{Corr}}
\newcommand{\badset}{BAD}
\newcommand{\transpose}{\mathrm{T}}
\newcommand{\numbkt}{B}
\newcommand{\capbkt}{c}

\section{Construction of Structured Seed PRG}
\label{sec:sprg}

In this section, we construct a family of structured-seed PRGs whose
evaluation has degree 2 in the private seed, and constant degree in
the public seed; the latter ensures that the computation on the public
seed lies in $\arithNC_0$ (which is exactly the class of functions computed by constant-degree polynomials).

\begin{theorem}
\label{thm:main}
  Let $\secparam$ be the security parameter. Let $d\in \mathbb{N}, \sparsity>0, \pstretch >1$ be arbitrary constants and $n=\poly(\secparam)$ be an arbitrary positive non-constant polynomial. 
  
  Then, assuming the following:
  \begin{itemize}
\item the existence of a constant locality Boolean $\prg$ with
  stretch $\pstretch >1$ and multilinear degree $\pdeg$ over $\Int$, and,
  \item $\flpn(\ell,n,r,p)$-assumption
  holds with respect to dimension
  $\ell = \displaystyle{ n^{1/\lceil \frac{\pdeg}{2} \rceil}}$, error
  rate $r = \ell^{-\sparsity}$,
  \end{itemize}
  there exists an $\sprg$ with polynomial
  stretch in $(\arithNC^0,\text{ deg } 2)$ that is
  $\gamma$-pseudorandom for every constant $\gamma > 0$.
Additionally, if both assumptions are secure against $2^{\secparam^{\nu}}$ time adversaries for some constant $\nu>0$, then, $\sprg$ is subexponentially
  $\gamma$-pseudorandom for every constant $\gamma > 0$.
\end{theorem}

\paragraph{Technical Overview.} Let $\prg=(\pidsamp,\peval)$ be the Boolean PRG with multilinear degree
$\pdeg$ and stretch $\pstretch$.  Our $\sprg$
will simply evaluate $\prg$ on an input $\vec \sigma \in \{0,1\}^n$ and return its output $\vec y \in \{0,1\}^{m}$ where $m=n^{\pstretch}$. The
challenge stems from the fact that the evaluation algorithm
$\peval_\pid(\vec \sigma)$ of $\prg$ has degree $\pdeg$ in its
private seed $\vec \sigma$, but the evaluation
algorithm $\peval'_\pid(\psd, \ssd)$ of $\sprg$ can only have degree 2
in the private seed $\ssd$. To resolve this, we pre-process
$\vec \sigma$ into appropriate public and private seeds $(\psd, \ssd)$
and leverage the LPN assumption over $\Int_p$ to show that the seed is
hidden.

Towards this, $\sprg$ ``encrypts'' the seed $\vec \sigma$ using LPN
samples over $\Int_p$ as follows:
\begin{align*}
  \text{Sample: } & \vec A \gets \Int_p^{\ell \times n}, \; 
                    \vec s  \gets \Int_p^{1\times \ell},\;  \vec{e} \leftarrow \cD^{1
                    \times n}_{r}(p)\\
  \text{Add to the function index $\pid'$: }  &  \vec A\\
  \text{Add to public seed $\psd$: }  &  \vec b = \vec s \vec A 
                                        + \vec e + \vec \sigma 
\end{align*}
It follows directly from the LPN over $\Int_p$ assumption that
$(\vec A, \vec b)$ is pseudorandom and hides $\vec
\sigma$. Furthermore, due to the sparsity of LPN noises, the vector
$\vec \sigma + \vec e$ differs from $\vec \sigma$ only at a
$r=\ell^{-\delta}$ fraction of components -- thus it is a sparsely
erroneous version of the seed.

Given such ``encryption'', by applying previous
techniques~\cite{C:AJLMS19,EC:JLMS19,JLS19,GJLS20} that work essentially by
``replacing monomials'' -- previous works replace monomials in the PRG seed with polynomials in the
LWE secret, and we here replace the monomials in the erroneous seed with polynomials in the LPN
secret -- we can compute $\prg$ on the erroneous seed
$\vec \sigma + \vec e$ via a polynomial $G^{(1)}$ (that depends on
$\vec A$) that has degree $\pdeg$ on the public component $\vec b$ and
only degree 2 on all possible degree $\lceil \frac{d}{2}\rceil$
monomials in $\vec s$. More precisely,
\begin{align}\label{eq:replace-var}
  \vec y'= \peval_\pid\big(\vec \sigma + \vec e \big) = G^{1}\big( \vec
  b\ ,\ (\ol{\vec s}^{\otimes \lceil \frac{d}{2}\rceil} ) \big), \qquad \ol{\vec s}  = \vec s||1
\end{align}

where $\vec v^{\otimes k}$ denotes tensoring the vector $\vec v$ with itself $k$
times, yielding a vector of dimension $\dim{\vec v}^k$. In particular,
observe that by setting the dimension $\ell$ of secret $\vec s$ to be
sufficiently small, the polynomial $G^{(1)}$ can be expanding; this is done by setting parameters $\ell(n)$ so that
$ \left(\ell^{\lceil \frac{d}{2}\rceil} + n\right)\ll
m(n)$. The reasoning behind  comparing the the number of output bits $m=n^{\pstretch}$ with the number of field elements in the seed of $\sprg$ is that if $m\gg \dim{(\vec b, \ol{\vec s}^{\otimes \lceil \frac{d}{2}\rceil})}$, then, we have polynomial expansion because the the length of the modulus $p$ is at most $\secparam$ bits which is asymptotically smaller than the parameter $n$.


However, the new problem is that even though the degree fits, $G^{(1)}$
only evaluates an erroneous output
$\vec y'= \peval_\pid(\vec \sigma+ \vec e)$, but we want to obtain the
correct output $\vec y= \peval_\pid(\vec \sigma)$. To correct errors,
we further modify the polynomial and include more pre-processed
information in the private seeds. Our key observation is the
following: Because LPN noises are sparse, and because $\peval_\pid$
has only constant locality, only a few outputs depend on erroneous
seed locations. We refer to them as $\bad$ outputs and let $\Cor$
denote the set of their indices. By a simple Markov argument, the
number of $\bad$ outputs is bounded by
$\thresh=m r \log n=\frac{m \log n}{\ell^{\sparsity}}$
with probability $1-o(1)$.
Leveraging this sparsity, $\sprg$ corrects $\mathsf{bad}$ outputs
using the method described below. In the low probability event where
there are too many $\bad$ outputs (greater than $\thresh$), it simply outputs 0.

We describe a sequence of ideas that lead to the final correction
method, starting with two wrong ideas that illustrate the difficulties we will overcome.
\begin{itemize}
\item The first wrong idea is correcting by adding the difference
  $\corr=\vec y-\vec y'$ between the correct and erroneous outputs,
  $\vec{y} =\peval_{\pid}(\vec \sigma)$ and
  $\vec y'=\peval_{\pid}(\vec \sigma + \vec e)$; refer to $\corr$ as
  the {\em correction vector}. To obtain the correct output,
  evaluation can compute the following polynomial
  $G^{(1)}\big( \vec b\ ,\ (\ol{\vec s}^{\otimes \lceil
    \frac{d}{2}\rceil} ) \big) + \corr$.
  The problem is that $\corr$ must be included in the seed, but it is
  as long as the output and would kill expansion.
  
\item To fix expansion, the second wrong idea is adding correction
  only for $\bad$ outputs, so that the seed only stores
  non-zero entries in $\corr$, which is short (bounded by $\thresh$
  elements). More precisely, the $j$'th output can be computed as
  $G^{(1)}_j\big( \vec b\ ,\ (\ol{\vec s}^{\otimes \lceil
    \frac{d}{2}\rceil} ) \big) + \corr_j$ if output $j$ is $\bad$
  and without adding $\corr_j$ otherwise.
  This fixes expansion, but now the evaluation polynomial depends on
  the location of $\bad$ outputs, which in turn leaks information
  of the location of LPN noises, and jeopardizes security.
\end{itemize}

The two wrong ideas illustrate the tension between the expansion and
security of $\sprg$. Our construction takes care of both, by {\em
  compressing} the correction vector $\corr$ to be polynomially
shorter than the output and stored in the seed, and {\em expanding} it
back during evaluation in a way that is oblivious of the location of
$\bad$ output bits. This is possible thanks to the sparsity of the
correction vector and the allowed degree 2 computation on the private
seed. Let's first illustrate our ideas in two simple cases.
\begin{description}
\item[Simple Case 1: Much fewer than $\sqrt m$ $\bad$ outputs.]
  Suppose hypothetically that the number of $\bad$ outputs is
  bounded by $z$ which is much smaller than $\sqrt m$. Thus, if we convert
  $\corr$ into a $\sqrt m \times \sqrt m$ matrix\footnote{Any
    injective mapping from a vector to a matrix that is efficient to
    compute and invert will do.}, it has low rank $z$. We can
  then factorize $\corr$ into two matrixes $\mat{U}$ and $\mat{V}$ of
  dimensions $\sqrt m \times z$ and $z \times \sqrt m$ respectively,
  such that $\corr = \mat U \mat V$, and compute the correct output as
  follows:
  \begin{equation*}
    \forall j \in [m],\ G_j^{(2)}\big( \vec b\ ,\ (\ol{\vec s}^{\otimes \lceil
    \frac{d}{2}\rceil},\ \mat U, \mat V) \big)=
    G^{(1)}_j\big( \vec b\ ,\ (\ol{\vec s}^{\otimes \lceil
    \frac{d}{2}\rceil} ) \big) + (\mat U \mat V)_{k_j,l_j}~, 
\end{equation*}
where $(k_j,l_j)$ is the corresponding index of the output bit $j$, in
the $\sqrt m \times \sqrt m$ matrix.
When $z \ll \sqrt m$, the matrices $\mat U, \mat V$ have 
$2 z\sqrt m$ field elements, which is polynomially smaller than $m=n^{\pstretch}$.  As
such, $G^{(2)}$ is expanding.

Moreover, observe that $G^{(2)}$ has only
degree 2 in the private seed and is completely oblivious of where the
$\bad$ outputs are.

  

\item[Simple Case 2: Evenly spread $\bad$ outputs.] The above
  method however cannot handle more than $\sqrt m$ $\bad$ outputs,
  whereas the actual number of $\bad$ outputs can be up to
  $\thresh=m(\log n)/\ell^{\sparsity}$, which can be much larger than
  $\sqrt m$ since $\sparsity$ is an arbitrarily small constant. Consider
  another hypothetical case where the $\bad$ outputs are evenly
  spread in the following sense: If we divide the matrix $\corr$ into
  $\frac{m}{\ell^\sparsity}$ blocks, each of dimension
  ${\ell^{\sparsity/2}}\times{\ell^{\sparsity/2}}$, there are at most
  $\log n$ $\bad$ outputs in each block. In this case, we can
  ``compress'' each block of $\corr$ separately using the idea from
  case 1. More specifically, for every block
  $i \in [\frac{m}{\ell^\sparsity}]$, we factor it into
  $\mat U_i \mat V_i$, with dimensions
  ${\ell^{\sparsity/2}} \times \log n$ and
  $\log n \times {\ell^{\sparsity/2}}$ respectively, and correct
  $\bad$ outputs as follows:
  \begin{align*}
    \forall j \in [m],\ G_j^{(2)}\left( \vec b\ ,\ \left(\ol{\vec s}^{\otimes \lceil
    \frac{d}{2}\rceil},\ \left(\mat U_i, \mat V_i\right)_{i \in [\frac{m}{\ell^\sparsity}]} \right)\right)=
    G_j^{(1)}\left( \vec b\ ,\ (\ol{\vec s}^{\otimes \lceil
    \frac{d}{2}\rceil} ) \right) + (\mat U_{i_j} \mat V_{i_j})_{k_j,l_j}~,
  \end{align*}
  where $i_j$ is the block that output $j$ belongs to, and
  $(k_j,l_j) \in [\ell^{\sparsity/2}]\times [\ell^{\sparsity/2}]$
  is its index within this
  block. 
  We observe that $G^{(2)}$ is expanding, since each matrix $\mat U_i$
  or $\mat V_i$ has ${\ell^{\sparsity/2}}\log n$ field elements,
  and the total number of elements is
  ${\ell^{\sparsity/2}}\log n \cdot \frac{m}{\ell^\sparsity}$
  which is polynomially smaller than $m$ as long as $\sparsity$ is
  positive.  Moreover, $G^{(2)}$ is oblivious of the location of $\bad$
  outputs just as in case 1. 
  
\end{description}
At this point, it is tempting to wish that $\bad$ outputs must be
evenly spread given that the LPN noises occur at random
locations. This is, however, not true because the input-output
dependency graph of $\prg$ is arbitrary, and the location of $\bad$
outputs are correlated. Consider the example that every output bit
of $\prg$ depends on the first seed bit. With probability
$\frac{1}{\ell^\sparsity}$ it is erroneous and so are all outputs.

To overcome this, our final idea is to ``force'' the even spreading of the $\bad$
outputs, by assigning them randomly into $\numbkt$ buckets, and then
compress the correction vector corresponding to each bucket.
\begin{description}
\item[Step 1: Randomly assign outputs.] We assign the outputs
  into $\numbkt$ buckets, via a random mapping
  $\phibkt: [m]\rightarrow [\numbkt]$.  The number of buckets is set to
  $\numbkt = \frac{mt}{\ell^\sparsity}$ where $t$ is a {\em slack
    parameter} set to $\secparam$.  By a Chernoff-style argument, we can
  show that each bucket contains at most $t^2\ell^{\sparsity}$
  output bits, and at most $t$ of them are $\bad$, except with
  negligible probability in $t$, which is also negligible in $\secparam$. As such,
  $\bad$ outputs are evenly spread among a small number of not-so-large buckets.

\item[Step 2: Compress the buckets.] Next, we organize each bucket $i$
  into a matrix $\mat M_i$ of dimension
  $t\ell^{\sparsity/2}\times t\ell^{\sparsity/2}$ and then compute its
  factorization $\mat M_i = \mat U_i \mat V_i$ with respect to
  matrices of dimensions $t\ell^{\sparsity/2}\times t$ and
  $t \times t\ell^{\sparsity/2}$ respectively. To form matrix
  $\mat M_i$, we use another mapping
  $\phiind : [m] \rightarrow [t\ell^{\sparsity/2}]\times
  [t\ell^{\sparsity/2}]$ to assign each output bit $j$ to an index
  $(k_j, l_j)$ in the matrix of the bucket $i_j$ it is assigned
  to. This assignment must guarantee that no two output bits in the
  same bucket (assigned according to $\phibkt$) have the same index;
  other than that, it can be arbitrary. $(\mat M_{i})_{k,l}$ is set to
  $\corr_j$ if there is $j$ such that $\phibkt(j)= i$ and
  $\phiind(j)= (k,l)$, and set to 0 if no such $j$ exists. Since every
  matrix $\mat M_i$ has at most $t$ non-zero entries, we can factor
  them and compute the correct output as:
  \begin{align*}
    \forall j \in [m],\ G_j^{(2)}\Big( \vec b\ ,\ \underbrace{\left(\ol{\vec s}^{\otimes \lceil
    \frac{d}{2}\rceil},\ \left(\mat U_i, \mat V_i\right)_{i \in [\numbkt]} \right)}_{\ssd}\Big)=
    G_j^{(1)}\left( \vec b\ ,\ (\ol{\vec s}^{\otimes \lceil
    \frac{d}{2}\rceil} ) \right) + (\mat U_{\phibkt(j)} \cdot \mat V_{\phibkt(j)})_{\phiind(j)}~,
  \end{align*}
  $G^{(2)}$ is expanding, because the number of field elements in $\mat U_i$'s and $\mat V_i$'s are much smaller than $m$,
  namely:
  $2 t^2\ell^{\sparsity/2}\numbkt =O( \frac{m \secparam^3}{\ell^{\sparsity/2}} )\ll m$. Note that it is important that the
  assignments $\phibkt$ and $\phiind$ are not included in the seed as
  their description is as long as the output. Fortunately, they are
  reusable and can be included in the function index
  $\pid'=(\pid, \vec A, \phibkt, \phiind)$.

\item[Step 3: Zeroize if uneven buckets.] Finally, to deal with the low
  probability event that some bucket is assigned more than
  $t^2\ell^{\sparsity}$ outputs or contains more than $t$ $\bad$
  outputs, we introduce a new variable called $\flag$. If either of the conditions above occur, our $\sprg$ sets $\flag=0$ and outputs zero. We then
  include $\flag$ in the public seed and augment the evaluation
  polynomial as follow:
  \begin{align*}
    \forall j \in [m],\
    \ G_j^{(3)}\Big(\underbrace{(\vec b,\ \flag)}_{\psd} ,\ \ssd \Big)
    = \flag \cdot G_j^{(2)}\left( \vec b,\ \ssd \right)~.
  \end{align*}
  This is our final evaluation polynomial. It has constant degree $d+1$ in the
  public seed $\psd$, degree 2 in the private seed $\ssd$, and
  expansion similar to that of $G^{(2)}$. For security, observe that
  the polynomial $G^{(3)}$ is independent of the location of LPN
  noises, while the public seed leaks 1-bit of information through
  $\flag$, which can be simulated efficiently via leakage
  simulation. Therefore, by the LPN over $\Int_p$ assumption, the seed
  $\vec \sigma$ of $\prg$ is hidden and the security of $\prg$ ensures
  that the output is pseudorandom when it is not zeroized. We now proceed to the formal construction and proof.
\end{description}

\paragraph{Construction.} We now formally describe our scheme.  Assume
the premise of the theorem. Let $(\pidsamp, \peval)$ be the function
index sampling algorithm and evaluation algorithm for the $\prg$. Recall that
its seed consists of only a private seed sampled uniformly and
randomly.


We first introduce and recall some notation. The construction is parameterized by
\begin{itemize}
\item $\secparam$ is the security parameter,
\item $n$ input length to the $\prg$. $n$ is arbitrary polynomial in $\secparam$,
\item the stretch $\pstretch$ and degree $\pdeg$ of $\prg$. Set $m=n^{\pstretch}$,
\item the LPN secret dimension $\ell = n^{1/\lceil d/2 \rceil}$,
  modulus $p$ be a $\secparam$ bit prime modulus,
\item a threshold $\thresh=\frac{m\cdot \log n}{\ell^{\sparsity}}$ of the
  number of $\bad$ outputs that can be tolerated,
\item a slack parameter $t$ used for bounding the capacity of and
  number of $\bad$ outputs in each bucket, set to
  $t=\secparam$.
\item a parameter $\numbkt = \frac{m\cdot t}{\ell^{\sparsity}}$ that
  indicates the number of buckets used.
\item a parameter $\capbkt = t^2 \ell^{\sparsity}$ that
  indicates the capacity of each bucket. 
\end{itemize}

\begin{description}
\item[$\pid' \leftarrow \pidsamp'(1^\secparam,1^{n'})$:] (Note that the PRG seed length $n$ below
  is an efficiently computable polynomial in $n'$, and can be inferred
  from the next seed sampling algorithm. See Claim \ref{claim:stretch} for the exact relationship between $n$ and $n'$.)\\ Sample
  $\pid \gets \pidsamp(1^\secparam,1^n)$ and
  $\vec A \leftarrow \Z^{\ell \times n}_p$. Prepare two functions
  $\vec \phi=(\phibkt, \phiind)$ as follows:
  \begin{itemize}
  \item Sample a random function $\phibkt:[m]\rightarrow [\numbkt]$ mapping
    every output location to one of $\numbkt$ buckets.  Let $\phibkt^{-1}(i)$ for
    $i \in [\numbkt]$ denote the set of preimages of $i$ through
    $\phibkt$. This set contains all outputs assigned to the
    bucket $i$.

  \item Prepare $\phiind:[m]\rightarrow [\sqrt \capbkt ]\times [\sqrt \capbkt]$ in two cases:
    \begin{itemize}
    \item {\em If some bucket exceeds capacity}, that is, there exists $i \in [\numbkt]$ such
      that $| \phibkt^{-1}(i) | > \capbkt$, set $\phiind$ to
      be a constant function always outputting $(1,1)$.
      
    \item {\em Otherwise if all buckets are under capacity}, for every
      index $j \in [m]$, $\phiind$ maps $j$ to a pair of indexes
      $(k_j,l_j) \in [\sqrt \capbkt]\times [\sqrt \capbkt]$, under the
      constraint that two distinct output bits $j_1 \ne j_2$ that are
      mapped into the same bucket $\phibkt(j_1) = \phibkt(j_2)$ must
      have distinct pairs of indices $\phiind(j_1)\neq \phiind(j_2)$.
    \end{itemize}
  \end{itemize}
  Output $\pid'=(\pid,\vec \phi,\vec A)$.

\item[$\sd \leftarrow \psdsamp'(\pid')$:] Generate the
  seed as follows: 
  \begin{itemize}
  \item Sample a $\prg$ seed $\vec \sigma \gets \{0,1\}^n$. 
  \item Prepare samples of LPN over $\Int_p$: Sample $\vec s  \gets \Int_p^{1\times \ell}$,  $\vec{e} \leftarrow \cD^{1
    \times n}_{r}(p)$, and set  
  \begin{align*}
    \vec b = \vec s \vec A +  {\color{black} \vec \sigma +  \vec e}~.
  \end{align*}
  \item Find indices $i \in [n]$ of seed bits where $\vec \sigma +  \vec e$ and $\vec
    \sigma$ differ, which are exactly these indices where $\vec e$ is
    not 0, and define:
    \begin{align*}
      \Err = \set{i \mid \sigma_i + e_i \ne \sigma_i} = \set{i \mid
      e_i \ne 0}~. 
    \end{align*}
    We say a seed index $i$ is {\em erroneous} if $i \in \Err$.  Since
    LPN noise is sparse, errors are sparse. 
    
  \item Find indices $j \in [m]$ of outputs that depend on one or more erroneous seed indices. Let $\vars_j$ denote the
    indices of seed bits that the $j$'th output  of $\eval_{\pid}$
    depends on. Define:
    \begin{align*}
      \Cor = \set{j \mid |\vars_j \cap \Err| \ge 1}~.
    \end{align*}
    We say an output index $j$ is $\bad$ if $j \in \Cor$, and  $\good$
    otherwise.
    
    \item Set $\flag=0$ if
      \begin{enumerate}
    \item {\em Too many $\bad$ output bits:} $\vert \Cor \vert > \thresh$, 
    \item {\em {\bf or} Some bucket exceeds capacity:} $\exists i \in [\numbkt]$,
      $\vert \phibkt^{-1}(i) \vert > \capbkt$, 
    \item {\em {\bf or} Some bucket contains too many $\bad$ outputs:}
      $\exists i \in [\numbkt]$, $\vert \phibkt^{-1}(i) \cap \Cor \vert> t$.
    \end{enumerate}
    Otherwise, set $\flag=1$.

  \item Compute the outputs of $\prg$ on input the correct seed and
    the erroneous seed, $\vec{y}=\peval_{\pid}(\vec{\sigma})$ and
    $\vec{y}'=\peval_{\vec{I}}(\vec{\sigma} + \vec{e})$. Set the
    correction vector $\corr = \vec y - \vec y'$.
    
  \item Construct matrices $\mat{M}_1, \dots, \mat{M}_\numbkt$, by
    setting
    \begin{align*}
      \forall j \in [m], \  \left(\mat{M}_{\phibkt(j)}\right)_{\phiind(j)}=\corr_j 
    \end{align*}
    Every other entry is set to $0$.

  \item ``Compress'' matrices $\mat{M}_1, \dots, \mat{M}_\numbkt$ as follows:
    \begin{itemize}
    \item If $\flag=1$, for every $i\in [\numbkt]$ compute
      factorization
    \begin{align*}
      \mat{M}_i=\mat{U}_i\mat{V}_i, \qquad \mat{U}_i \in \Z^{\sqrt
      \capbkt \times t }_p, \ \mat{V}_i \in \Z^{t\times \sqrt
      \capbkt }_p
    \end{align*}
    This factorization exists because when $\flag = 1$,
    condition 3 above implies that each $\mat{M}_i$ has at most $t$
    nonzero entries, and hence rank at most $t$.

  \item If $\flag=0$, for every $i\in [\numbkt]$, set $\mat{U}_i$ and
    $\mat{V}_i$ to be $0$ matrices.
    \end{itemize}
        
  \item Set the public seed to
    \begin{align*}
      \psd = ( \vec b, \flag)~.
    \end{align*}
    
  \item Prepare the private seed $\ssd$ as follows. Let
    $\ol{\vec s} = \vec s||1$.
    \begin{align}
    \label{eq:ssd}
      \ssd= \left (\ol{\vec{s}}^{\otimes \lceil \frac{d}{2}\rceil}, \left \{ \mat{U}_i ,\mat{V}_i \right\}_{i\in [\numbkt]} \right)
    \end{align}

  Output $\sd = (\psd, \ssd)$ as $\Int_p$
  elements.

\end{itemize}

\item[$\vec y \rightarrow \peval'(\pid', \sd)$:] Compute
  $\vec y \leftarrow \peval(\pid, \vec \sigma)$, and output
  $\vec z=\flag \cdot \vec y$. This computation is done via a
  polynomial $G_{\pid'}^{(3)}$ described below that has constant degree $\pdeg+1$ in
  the public seed and only degree 2 in the private seed, that is, 
  \begin{align*}
    \peval'(\pid',\sd)= \flag \cdot \vec y = \flag \cdot  \peval_\pid(\vec \sigma) = G_{\pid'}^{(3)}(\psd, \ssd)~.
  \end{align*}
  We next define $G_{\pid'}^{(3)}$ using intermediate polynomials
  $G^{(1)}_{\pid'}$ and $G^{(2)}_{\pid'}$. For simplicity of notation,
  we suppress subscript $\pid'$ below.
  \begin{itemize}
  \item Every output bit of $\peval$ is a linear combination of degree
    $\pdeg$ monomials (without loss of generality, assume that all
    monomials have exactly degree $\pdeg$ which can be done by
    including 1 in the seed $\vec \sigma$).
  \begin{notation}
    Let us introduce some notation for monomials. A monomial $\mnl$
    on a vector $\vec a$ is represented by the set of indices
    $\mnl=\set{i_1, i_2, \dots, i_k}$ of variables used in it. $\mnl$
    evaluated on $\vec a$ is $\prod_{i \in \mnl} a_i$ if
    $\mnl \ne \emptyset$ and 1 otherwise. We will use the notation
    $a_\mnl = \prod_{i \in \mnl} a_i$. 
    We abuse notation to also use a polynomial $g$ to denote
    the set of monomials involved in its computation; hence
    $\mnl \in g$  says monomial $\mnl$ has a nonzero coefficient in $g$.
  \end{notation}
   With the above notation, we can write $\peval$ as
  \begin{align*}
    \forall j \in [m], \ \ y_j =\peval_j(\vec \sigma) =
    L_j((\sigma_\mnl)_{\mnl \in \peval_j}) \ , \text{ for a linear
     $L_j$  }.
  \end{align*}
  

\item $(\vec A, \vec b = \vec s \vec A + \vec x)$ in the public seed
  encodes $\vec x=\vec{\sigma} + \vec{e}$. Therefore, we can compute every monomial
  $x_{\mnlv}$ as follows:
  \begin{align*}
    x_i & = \innerp{ \vec c_i,\ \ol{\vec s}} & \vec c_i = -\vec  a^{\transpose}_i || b_i,
    \ \vec a_i \text{ is the $i$th column of  $\vec A$}\\
    x_\mnlv & = \innerp{ \otimes_{i \in \mnlv} \vec c_i,\ \otimes_{i
                  \in \mnlv} \ol{\vec  s}} & 
  \end{align*}
  (Recall that
  $\otimes_{i \in \mnlv} \vec z_i = \vec z_{i_1} \otimes \dots \otimes
  \vec z_{i_k}$ if $\mnlv = \set{i_1, \dots, i_k}$ and is not empty;
  otherwise, it equals 1.) 
  Combining with the previous step, we obtain a polynomial
  $G^{(1)}(\vec b, \ssd)$ that computes $\peval(\vec \sigma + \vec
  e)$: 
  \begin{align}
    G_j^{(1)}( \vec b,\ssd) \coloneqq  L_j\left(\left (\innerp{ \otimes_{i \in \mnlv} \vec c_i,\ \otimes_{i
    \in \mnlv} \ol{\vec  s}}\right )_{\mnlv \in
    \peval_j}\right)~. \label{eq:G1} 
  \end{align}
  Note that $G^{(1)}$, by which we mean $G^{(1)}_{\pid'}$, implicitly
  depends on $\vec A$ contained in $\pid'$.  Since all relevant
  monomials $\mnlv$ have  degree $\pdeg$, we have that $G^{(1)}$ has degree at
  most $\pdeg$ in $\psd$, and degree 2 in $\ssd$. The latter follows
  from the fact that $\ssd$ contains
  $\ol{\vec{s}}^{\otimes \lceil \frac{d}{2}\rceil}$ (see
  Equation~\eqref{eq:replace-var}), and hence $\ssd\otimes \ssd$
  contains all monomials in ${\vec s}$ of total degrees $\pdeg$.

  Since only $\bad$ outputs depend on erroneous seed bits such
  that $\sigma_i + e_i \ne \sigma_i$, we have that the output of
  $G^{(1)}$ agrees with the correct output
  $\vec y = \peval(\vec \sigma)$ on all $\good$ output bits.
  \begin{align*}
    \forall j \not\in \Cor,\  \peval_j(\vec \sigma) = G_j^{(1)}( \vec b, \ssd)~.
  \end{align*}
  
\item To further correct $\bad$ output bits, we add to $G^{(1)}$ all
  the expanded correction vectors as follows:
  \begin{align*}
      G^{(2)}_j(\psd , \ssd)  \coloneqq  G^{(1)}_j( \vec b,\ssd) +
                               \left(\mat{U}_{\phibkt(j)} \mat{V}_{\phibkt}(j)\right)_{\phiind(j)}
                             = G^{(1)}_j( \vec b,\ssd) + \left(\mat{M}_{\phibkt(j)}\right)_{\phiind(j)} ~.
      \end{align*}
      We have that $G^{(2)}$ agrees with the correct output
      $\vec y = \eval(\vec \sigma)$ if $\flag = 1$. This is because
      under the three conditions for $\flag =1$, every entry $j$ in
      the correction vector $\corr_j$ is placed at entry
      $\left(\mat M_{\phibkt(j)}\right)_{\phiind(j)}$. Adding it back
      as above produces the correct output.

    Observe that the function is quadratic in $\ssd$ and  degree $\pdeg$ in the public component of the seed $\psd$.

  \item When $\flag = 0$, however, $\sprg$ needs to output all
    zero. This can be done by simply multiplying $\flag$ to the output of
    $G^{(2)}$, giving the final
    polynomial
    \begin{align}
      G^{(3)}(\psd, \ssd) = \flag \cdot G^{(2)}(\psd, \ssd)~.
    \end{align}
    At last, $G^{(3)}$ has  degree $\pdeg + 1$ in the public seed, and only degree 2
    in the private seed, as desired. 
  \end{itemize}
\end{description}
  
\paragraph{Analysis of Stretch.}
We derive a set of constraints, under which $\sprg$ has polynomial
stretch. Recall that $\prg$ output length is $m = n^\pstretch$, degree
$\pdeg$, LPN secret dimension $\ell = n^{1/\lceil d/2\rceil}$, modulus
$p = O(2^{\secparam})$, and the slack parameter $t=\secparam$. 

\begin{claim}
\label{claim:stretch}
 For the parameters as set in the Construction, $\sprg$ has stretch of $\pstretch'$ for some constant $\pstretch'>1$.
 \end{claim}
\begin{proof}
Let's start by analyzing the length of the public and
private seeds. 
\begin{itemize}
\item The public seed contains $\psd = (\vec b, \flag)$ and
  has bit length
\begin{align*}
   \ \blen{\psd} = O(n \log p) =
    O(n\cdot \secparam)~.
\end{align*}

\item The private seed $\ssd$ contains $\ssd_1, \ssd_2$ as follows:
   \begin{align*}
    \label{eq:ssd1}
       \ssd_1 =  \ol{\vec s}^{\otimes \lceil \frac{\pdeg}{2} \rceil},\qquad
     \ssd_2 = \left \{ \mat{U}_{i}, \mat{V}_i \right \}_{i\in [\numbkt]}.
    \end{align*}
The bit-lengths are:
  \begin{align*}
    \blen{ \ssd_1}  = & (\ell+1)^{\lceil d/2 \rceil}\log  p & \\
    = & O\left(n^{\frac{1}{\lceil d/2
        \rceil}}\right)^{\lceil d/2 \rceil}\log  p =
        O(n\cdot \secparam)   & \text{by } \ell = n^{\lceil d/2
        \rceil}, \ \log  p =\secparam \\ 
    \blen{\ssd_2} = & 2 \numbkt\cdot t\cdot \sqrt \capbkt\cdot \log  p
   &\\
    = & \frac{2m t}{\ell^{\sparsity}} \cdot t \cdot t \ell^{\sparsity/2}\cdot
        \log  p
        = \frac{2m t^3\log p}{\ell^{\sparsity/2}}  &
    \text{by }\numbkt= \frac{m t}{\ell^{\sparsity}},\ \capbkt = t^2\ell^\sparsity \\
    = & \frac{2m \secparam^4} {\ell^{\sparsity/2}} & \text{by }t=\secparam\\
  \end{align*} 
\end{itemize} 
\vspace{-3ex}
Because $\ell^{\sparsity/2} = n^{\frac{\sparsity}{2\lceil\frac{\pdeg}{2}\rceil}}$ and $m=n^\pstretch $, we have:
\begin{align*}
    \blen{\sd}= \blen{\psd}+\blen{\ssd_1}+\blen{\ssd_2} = O((n+ n^{\pstretch -\frac{\sparsity}{2\lceil \frac{\pdeg}{2}\rceil}})\cdot \secparam^4)
\end{align*}
We set $n'=O(n+ n^{\pstretch -\frac{\sparsity}{2\lceil \frac{\pdeg}{2}\rceil}})$, therefore $m={n'}^{\pstretch'}$ for some $\pstretch'>1$. This concludes the proof.
\end{proof}

\paragraph{Proof of Pseudorandomness} 
We prove the following proposition which  implies that
$\sprg$ is $\gamma$-pseudorandom for any constant $\gamma$.

\begin{proposition} \label{propo:sprg-sec}
Let $\ell, n, r, p$ be defined as above.  
For any running time $T = T(\secparam) \in \mathbb{N}$, if
\begin{itemize}
\item $\flpn(\ell,n,r,p)$ is
  $(T,\epsilon_{\flpn})$-indistinguishable for  advantage
  $\epsilon_{\flpn} = o(1)$, and
\item $\prg$ is $(T, \epsilon_{\prg})$-pseudorandom for  advantage
  $\epsilon_{\prg} = o(1)$,
\end{itemize}
$\sprg$ satisfies that for every constant $\gamma \in (0,1)$, the following two distributions are $(T,\gamma)$-indistinguishable.
  \begin{eqnarray*}
   \Big\{\ (\pid,\vec \phi ,\vec A,\vec b, \flag, \vec z) & : &   (\pid,\ \vec \phi, \ \vec A) \gets
    \pidsamp'(1^{n'}), \ (\psd,\ssd) \gets \psdsamp'(\pid'),\ \vec z \gets \peval'(\pid, \sd) \Big\} \\
     \Big\{\ (\pid, \vec \phi,\vec A , \vec b,  \flag, \vec r)  & : &
       (\pid,\ \vec \phi, \ \vec A) \gets \pidsamp'(1^{n'}),\
       (\psd,\ssd) \gets \psdsamp'(\pid'),\
       \vec r \gets \zo^m 
     \Big\},                                     
  \end{eqnarray*}
  (Recall that  $P = (\vec b, \flag)$.)
\end{proposition}
We start with some intuition behind the proposition. Observe first
that if $\flag$ is removed, the above two distributions becomes truly
indistinguishable. This follows from the facts that i) $\pid$ and
$\vec \phi$ are completely independent of $(\vec A, \vec b, \vec z)$
or $(\vec A, \vec b, \vec r)$, and ii) $(\vec A, \vec b, \vec z)$ and
$(\vec A, \vec b, \vec r)$ are indistinguishable following from the
LPN over $\Int_p$ assumption and the pseudorandomness of $\prg$. The
latter indistinguishability is the heart of the security of $\sprg$,
and is captured in Lemma~\ref{lem:flpnprg} below. Towards the
proposition, we need to additional show that publishing $\flag$ does
not completely destroy the indistinguishability. This follows from the
facts that i) $\flag$ is only 0 with sub-constant probability, and ii)
it can be viewed as a single bit leakage of the randomness used for
sampling the rest of the distributions, and can be simulated
efficiently by the leakage simulation lemma,
Theorem~\ref{thm:leakage}. The formal proof of the proposition below
presents the details.

\begin{lemma}
\label{lem:flpnprg}
Let $G:\{0,1\}^{1 \times n}\rightarrow \{0, 1\}^{1\times m(n)}$ be a
$(T,\epsilon_{\prg})$-secure pseudorandom generator. Assume that
$\flpn(\ell,n,r,p)$ is $(T,\epsilon_{\flpn})$-secure. Then the following two distributions are $(T,\epsilon_{\flpn}+\epsilon_{\prg})$-indistinguishable:
\begin{align*}
\cD_{1}   & = \Big\{(\vec{A}, \; \vec{b}= \vec{s}\cdot \vec{A}+ \vec{e} +
            \vec{\sigma},\; G(\vec{\sigma}))  \ : \ \vec{A}\leftarrow \Z^{\ell\times n}_{p};\;
\vec{s}\leftarrow \Z^{1\times \ell}_{p};\;
\vec{e} \leftarrow \cD^{1 \times n}_{r}(p);\;
\vec{\sigma}\leftarrow \{0,1\}^{1 \times n}\Big \}\\
\cD_{2}  & = \Big \{ (\vec{A},\; \vec{u},\; \vec{w})\ : \  \vec{A}\leftarrow \Z^{\ell\times n}_{p};\;
\vec{u}\leftarrow \Z^{1\times n}_{p};\; \vec{w} \leftarrow \{0,1\}^{1\times m(n)} \Big \}
\end{align*}
\end{lemma}

\begin{proof}
 We introduce one intermediate distribution $\cD'$ defined as follows: 
 \begin{align*}
\cD'  = \Big\{ (\vec{A},\; \vec{u},\; G(\vec \sigma)) \ : \ \vec{A}\leftarrow \Z^{\ell\times n}_{p};\;
\vec{u}\leftarrow \Z^{1\times n}_{p};\;
\vec \sigma \leftarrow \{0,1\}^n\Big\}
\end{align*}
Now observe that $\cD'$ is $(T,\epsilon_{\flpn})$-indistinguishable to
$\cD_{1}$ following immediately from the
$(T,\epsilon_{\flpn})$-indistinguishability of the $\flpn(\ell,n,r,p)$
assumption.  Finally, observe that $\cD'$ is
$(T,\epsilon_{\prg})$-indistinguishable to $\cD_2$ due to
$(T,\epsilon_{\prg})$-security of $G$. Therefore, the lemma holds.
\end{proof}

\begin{proof}[Proof of Proposition~\ref{propo:sprg-sec}]
  We now list a few hybrids $\hyb_0, \hyb_1, \hyb_2, \hyb_3$, where
  the first one corresponds to the first distribution in the
  proposition, and the last one corresponds to the second distribution
  in the proposition. We abuse notation to also use $\hyb_i$ to denote
  the output distribution of the hybrid.  Let $\gamma$ be the claimed
  advantage of the adversary $\A$, running in time
  $T q(\secparam)$ for a polynomial $q$. Let $\cD_{\phi,\pid}$ denote the the
  distribution that samples the functions $\vec \phi$. 
  
  \begin{description}
  \item[Hybrid $\hyb_0$] samples $(\pid',\psd, \vec y)$ honestly as in
    the first distribution, that is,
    \begin{align*}
      \text{Sample: } & \vec A \gets \Int_p^{\ell \times n}, \; 
                        \vec s  \gets \Int_p^{1\times \ell},\;  \vec{e} \leftarrow \cD^{1
                        \times n}_{r}(p) ,\ \vec \sigma \gets \{0,1\}^n  \\
                        &
                        \; \pid \gets \pidsamp(1^{\secparam},1^n),\ \vec y= \peval_{I}(\vec \sigma),\ \vec \phi \gets \cD_{\phi,\pid}  \\
      \text{Output: }   
      &\pid, \; \vec{\phi},\ \vec A, \vec b = \vec s \vec A + \vec e + \vec \sigma,\ \flag \cdot  \vec y   \\
   & \text{where }\flag = 1 \text{ iff:}\\
   &  \text{1) } \vert  \Cor \vert \leq \thresh \text{ and,}\\
   &  \text{2) } \forall i \in [\numbkt], \ \vert \phibkt^{-1}(i)\cap \Cor\vert \leq t \  \text{ and,}\\
    &  \text{3) }\forall i \in [\numbkt], \  \vert \phibkt^{-1}(i) \vert \leq \ell^{\sparsity}\cdot t^2. 
    \end{align*}
    Note that the value of $\flag$ is correlated with that of
    $(\pid,\vec \phi, \vec A, \vec b, \vec y)$. Therefore, $\flag$ can be viewed
    as a single-bit leakage of the randomness used for sampling
    $(\pid, \vec \phi,\vec A, \vec b, \vec y)$.

  \item[Hybrid $\hyb_1$] instead of generating $\flag$ honestly, first
    samples $X=(\pid, \vec{\phi}, \vec A, \vec b, \vec y)$ honestly, and then
    invokes the leakage simulation lemma, Lemma~\ref{thm:leakage}, to
    simulate $\flag$ using $X$, for $Tq(\secparam) + \poly(\secparam)$ time
    adversaries with at most $\frac{\gamma}{3}$ advantage. Let $\Sim$
    be the simulator given by Theorem~\ref{thm:leakage}. 
    \begin{align*}
      \text{Sample: } & \vec A \gets \Int_p^{\ell \times n}, \; 
                        \vec s  \gets \Int_p^{1\times \ell},\;  \vec{e} \leftarrow \cD^{1
                        \times n}_{r}(p), \; \vec \sigma \gets \{0,1\}^n, \\
                        &
                        \; \pid \gets \pidsamp(1^{\secparam},1^n),\ \vec y= \peval_{I}(\vec \sigma),\ \vec{\phi}\gets \cD_{\phi,\pid} \\
      \text{Output: }   
      &\pid, \; \vec \phi,\ \vec A, \ \vec b = \vec s \vec A + \vec e + \vec \sigma,\ \flag \cdot  \vec y  \\
   & \text{where }\underline{\color{red} \flag = \mathsf{Sim}(I,\vec \phi ,\vec A, \vec b, \vec y)} 
    \end{align*}
    The leakage simulation lemma guarantees that the running time of
    $\Sim$ is bounded by
    $O((Tq(\secparam) + \poly(\secparam))\cdot \frac{9}{\gamma^2}\cdot
    2^1)=Tq'(\secparam))$ for a fixed polynomial $q'$, and $\A$ cannot
    distinguish $\hyb_0$ from $\hyb_1$ with advantage more than
    $\frac{\gamma}{3}$.
    \begin{claim}\label{claim:hyb1}
      For any adversary $\A$ running in time $Tq(n)$ for some
      polynomial $q$, 
      \begin{align*}
        \vert \Pr[\A(\hyb_0)=1]- \Pr[\A(\hyb_1)=1] \vert \leq \frac{\gamma}{3}~.
      \end{align*}
      Furthermore, the running time of $\mathsf{Sim}$ is $Tq'(\secparam)$ for
      some polynomial $q'$.
    \end{claim}
    
    This claim is immediate from Lemma~\ref{thm:leakage}.
    
    
  \item[Hybrid $\hyb_2$] samples $\vec{A},\vec{b}$ and $\vec y$
    uniformly and randomly.
    \begin{align*}
      \text{Sample: } & \underline{\color{red} \vec A \gets \Int_p^{\ell \times n}, \; 
                        \vec b \leftarrow \Z^{1\times n}_p} \; \\
                        &
                         \pid \gets \pidsamp(1^{\secparam},1^n),\
                          \underline{\color{red} \vec y \leftarrow \zo^m},\ \vec \phi \gets \cD_{\phi,\pid} \\
      \text{Output: }   
      &\pid, \;  \vec \phi,\ \vec A, \ \vec b,\ \flag \cdot  \vec y  \\
   & \text{where }\flag = \mathsf{Sim}(\pid, \vec \phi,\vec A, \vec b, \vec y) 
    \end{align*}
    Lemma~\ref{lem:flpnprg} shows that $(\vec A, \vec b, \vec y)$
    generated honestly as in $\hyb_1$ and $(\vec A, \vec b, \vec y)$
    sampled all at random as in $\hyb_2$ are indistinguishable, due to
    the $\flpn$ assumption and the pseudorandomness of $\prg$.  Here
    the adversary $\A$ runs in time $T q(\secparam)$ and the simulator $\Sim$
    runs in time $Tq'(\secparam)$ time, for polynomials $q, q'$. Thus, we get
    \begin{claim}\label{claim:hyb2}
      For any adversary $\A$, running in time $T$, if
      $\flpn(\ell,n,r,p)$ is $(T,\epsilon_{\flpn})$-secure and
      $\prg$ satisfies $(T,\epsilon_{\prg})$-pseudorandomness, then,
      \begin{align*}
    \vert \Pr[\A(\hyb_1)=1]- \Pr[\A(\hyb_2)=1] \vert \leq \epsilon_{\prg}+\epsilon_{\flpn}
      \end{align*}
    \end{claim}
    
    This claim follows immediately from Lemma~\ref{lem:flpnprg}.
    
    \item[Hybrid $\hyb_3$] no longer generates $\flag$ and simply
      outputs the random string $\vec y$ instead of $\flag\cdot\vec
      y$. 
    \begin{align*}
      \text{Sample: } & \vec A \gets \Int_p^{\ell \times n}, \; 
                        \vec b \leftarrow \Z^{1\times n}_p \; \\
                        &
                         \pid \gets \pidsamp(1^{\secparam},1^n),\ \vec y \leftarrow \zo^m,\ \vec{\phi}\gets \cD_{\phi,\pid} \\
      \text{Output: }   
      & \pid, \vec \phi,\ \vec A,\ \vec b, \underline{\color{red} \vec y} 
    \end{align*}
    Observe that $\hyb_2$ and $\hyb_3$ are only distinguishable when
    $\flag = 0$ in $\hyb_2$. By bounding the probability of
    $\flag = 0$ in $\hyb_2$, we can show that  
    \begin{claim}\label{claim:hyb3}
      For any adversary $\A$, 
      \begin{align*}
        \vert \Pr[\A(\hyb_2)=1]- \Pr[\A(\hyb_3)=1] \vert \leq \frac{\gamma}{2}
      \end{align*}
    \end{claim}
    The formal proof of this lemma is provided below. 
  \end{description}
  Combining the hybrids above, we conclude that $\A$ cannot
  distinguish $\hyb_0$ and $\hyb_3$ with advantage more than
  $\frac{5\cdot \gamma}{6}+\epsilon_{\prg}+\epsilon_{\flpn} < \gamma$,
  which gives a contradiction. Therefore, the indistinguishability
  stated in the proposition holds. We now complete the final remaining
  piece -- the proof of Claim~\ref{claim:hyb3}.
  
  \begin{proof}[Proof of Claim~\ref{claim:hyb3}]
    This indistinguishability is statistical. We start with showing
    that the probability that $\flag = 0$ in $\hyb_0$ is
    $O(\frac{1}{\log n})$.
    Towards this, we bound probability of all three conditions for setting $\flag=0$ and then apply the union bound. 
    
    \begin{itemize}
    \item $\Pr[\vert \Cor \vert >\thresh] \leq O(\frac{1}{\log
        n})$. Observe that by the fact that $\peval_\pid$ has constant
      locality  in $\vec \sigma$, the probability that any
      single output bit $j \in[m]$ is $\bad$ is bounded by
      $O(r) = \frac{O(1)}{\ell^{\sparsity}}$, where $r$ is the rate of
      LPN noises. Therefore, the expected number of $\bad$
      output bits is
    \begin{align*}
      \E[ |\Cor| ] =  \frac{O(1)m}{\ell^{\sparsity}}
    \end{align*}
    Thus by Markov's inequality,
    \begin{align*}
      \Prob[ |\Cor| > \thresh] \le \frac{1}{\thresh} \cdot
      \frac{O(1)m}{\ell^{\sparsity}\cdot \thresh} =  \frac{O(1)}{\log n}~.
    \end{align*}
   The last equality follows from the fact that $\thresh=\frac{m\cdot
     \log n}{\ell^{\sparsity}}$.
   
 \item For any $i\in [\numbkt]$,
   $\Pr_{\phibkt}\left [\vert \phibkt^{-1}(i)\cap \Cor\vert > t \mid
     \vert \Cor \vert \leq \thresh \right] \leq \negl(n)$. Suppose
   $\vert \Cor \vert =\thresh'$ where $\thresh'\leq \thresh$, and
   since $\phibkt: [m]\rightarrow [\numbkt]$ is a random function, we
   have:
   \begin{align*}
    \Pr_{\phibkt}\left [\vert \phibkt^{-1}(i)\cap \Cor\vert > t \mid
     \vert \Cor \vert = \thresh' \right] &\leq {\thresh'\choose t}
                                           \cdot \frac{1}{\numbkt^{t}}
     & \\
    & \leq \left (\frac{e \cdot \thresh'}{t} \right )^t\cdot
      \frac{1}{\numbkt^t} & \text{ by Stirling's approximation} \\
   & \leq  \left (\frac{e}{t} \right)^{t} \leq e^{-t}
     & \text{ by }\thresh' < \thresh < \numbkt\\
     & = \negl(\secparam) & \text{ by } t = \secparam
   \end{align*}
   
   \item For any $i\in \numbkt$, $\Pr_{\phibkt}[\vert \phibkt^{-1}(i) \vert > \ell^{\sparsity}\cdot t^2] \leq \negl(\secparam)$. Since $\phibkt$ is a random function,
   \begin{align*}
    \Pr_{\phibkt}\left [\vert \phibkt^{-1}(i) \vert >t^{2}\cdot
     \ell^{\sparsity} \right]
     &\leq {m\choose \ell^{\sparsity}\cdot t^2} \cdot \left (\frac{1}{\numbkt}\right)^{\ell^{\sparsity}\cdot t^2} & \\
    & \leq \left (\frac{e \cdot m}{\ell^{\sparsity}\cdot t^2} \right
      )^{\ell^{\sparsity}\cdot t^2}\cdot \left
      (\frac{1}{\numbkt}\right )^{\ell^{\sparsity}\cdot t^2} & \text{ by Stirling's approximation} \\
   & = \left(\frac{e\cdot m }{\numbkt\cdot \ell^{\sparsity}\cdot t^2}
     \right )^{\ell^{\sparsity}\cdot t^2} \leq \left(\frac{1}{t^2}\right)^{\ell^{\sparsity}\cdot t^2} & \text{ by } \numbkt =  \frac{m t}{\ell^\sparsity} > \frac{em}{\ell^\sparsity}\\
    &\leq t^{-2t^2 } = \negl(\secparam)& \text{ by } \ell^\sparsity >1 \text { and }  t = \secparam
   \end{align*}
    \end{itemize}
    
     Applying the three observations above, from a union bound it
     follows that $\Pr[\flag=0]=O(\frac{1}{\log n})$.

    Next, for adversaries of run time $Tq(\secparam)$,
    Claim~\ref{claim:hyb1} shows that $\hyb_0$ and $\hyb_1$ cannot be
    distinguished with advantage more than $\frac{\gamma}{3}$, and
    Claim~\ref{claim:hyb2} shows that $\hyb_1$ and $\hyb_2$ cannot be
    distinguished with advantage more than
    $ \epsilon_{\prg}+\epsilon_{\flpn}$, which is sub-constant. Therefore, the probability
    that $\flag = 0$ in $\hyb_2$ is upper bounded by
    \begin{align*}
      \Prob[ \flag = 0 \text{ in } \hyb_2] \le \frac{O(1)}{\log n} +
      \frac{\gamma}{3} + \epsilon_{\prg}+\epsilon_{\flpn} \le \frac{\gamma}{2}~.
    \end{align*}
     
    Finally, we upper bound the statistical distance between
    $\hyb_2$ and $\hyb_3$, which is
     \begin{align*}
       SD(\hyb_2, \hyb_3) = \frac{1}{2}\cdot
       \sum_{(I,\vec \phi,\vec{A},\vec{b},\vec{y})} \Big\vert \Pr[\hyb_2 =
       (I,\vec \phi,\vec{A},\vec{b},\vec{y})]-\Pr[\hyb_3 = (I,\vec \phi,\vec{A},\vec{b},\vec{y})] \Big\vert~.
     \end{align*}
     For $b\in \{0,1\}$, let $F_b$ be the set of tuples
     $(I,\vec{A},\vec{b},\vec{y})$ that generate $\flag = b$ through
     $\Sim$,
     \begin{align*}
       F_b = \left\{(I,\vec \phi,\vec{A},\vec{b},\vec{y}) \mid \mathsf{Sim}((I,\vec \phi,\vec{A},\vec{b},\vec{y})=b \right\}~.
     \end{align*}
     Then, we have:
     \begin{align*}
       SD(\hyb_2, \hyb_3) &= \frac{1}{2}\cdot
       \sum_{(I,\vec \phi,\vec{A},\vec{b},\vec{y}) \in F_0} \Big\vert \Pr[\hyb_2 =
       (I,\vec \phi,\vec{A},\vec{b},\vec{y})]-\Pr[\hyb_3 =
                            (I,\vec \phi,\vec{A},\vec{b},\vec{y})] \Big\vert~\\
       & + \frac{1}{2}\cdot
       \sum_{(I,\vec \phi,\vec{A},\vec{b},\vec{y}) \in F_1} \Big\vert \Pr[\hyb_2 =
       (I,\vec \phi,\vec{A},\vec{b},\vec{y})]-\Pr[\hyb_3 =
         (I,\vec \phi,\vec{A},\vec{b},\vec{y})] \Big\vert\\
       & = \frac{1}{2}\cdot
       \sum_{(I,\vec \phi,\vec{A},\vec{b},\vec{y}) \in F_0} \Big\vert \Pr[\hyb_2 =
       (I,\vec \phi,\vec{A},\vec{b},\vec{y})]-\Pr[\hyb_3 =
                            (I,\vec \phi,\vec{A},\vec{b},\vec{y})] \Big\vert~\\
       &
       \leq  \Pr[\flag = 0 \text{ in } \hyb_2] \le \frac{\gamma}{2}
     \end{align*}
     where the second equality follows from the fact that in $\hyb_2$
     and $\hyb_3$ the probability of outputing a tuple
     $(\pid', \vec \phi,\vec A, \vec b, \vec y)$ that belongs to $F_1$, or
     equivalently generates $\flag = 1$ via $\Sim$, is the same.
     This concludes the claim. 
    
   \end{proof}
   
\end{proof}


\newcommand{\ilen}{l}
\newcommand{\olen}{m}
\newcommand{\fedeg}{D}

\section{Bootstrapping to Indistinguishability Obfuscation}
\label{sec:iotheorems}
We now describe a pathway to $\iO$ and 
$\fe$ for all circuits.

  

\paragraph{From Structured-Seed PRG to Perturbation Resilient Generator.}
Starting from structured-seed $\prg$, we show how to construct
perturbation resilient generators, denoted as $\drg$. $\drg$ is the
key ingredient in several recent $\iO$ constructions
\cite{C:AJLMS19,EC:JLMS19,JLS19}. Roughly speaking, they have the same
syntax as structured-seed $\prg$s with the notable difference that it
has integer outputs $\vec y$ of polynomial magnitude; further, they
only satisfy weak pseudorandomness called {\em perturbation
  resilience} guaranteeing that $\vec y + \vec \beta$ for arbitrary adversarially chosen small integer vector $\vec \beta$ is weakly indistinguishable from $\vec y$
itself. The formal definition of $\drg$ is provided in
Definition~\ref{def:drg} in Section~\ref{sec:drg}.

\begin{theorem}[$\sprg$ to $\drg$, proven in Section~\ref{sec:drg}]
\label{thm:prgdrg}
  Let $\secparam\in \mathbb{N}$ be the security parameter,
  $\gamma\in (0,1)$, and $\pstretch > 1$.  Assume the existence of a
  (subexponentially) $\gamma$-pseudorandom $\sprg$ in
  $(\mathbb{C},\deg\ d)$ with stretch $\pstretch$. For any
  constant $0<\pstretch'<\pstretch$, there exists a (subexponentially)
  $(2\gamma+O(\frac{1}{\secparam}))$-perturbation resilient $\drg$ in
  $(\mathbb{C},\deg \ d)$ with a stretch $\pstretch'$.
\end{theorem}

\paragraph{From  Perturbation Resilient  Generator to Weak FE for $\NC^0$.}
It was shown in~\cite{C:AJLMS19,EC:JLMS19,JLS19} that $\drg$, along
with $\mathsf{SXDH}$, $\mathsf{LWE}$ and $\prg$ in $\mathsf{NC}^0$,
can be used to construct a secret-key functional encryption scheme for
$\NC^0$ circuits. The $\fe$ scheme supports only a {\em
  single secret key} for a function with multiple output bits, has
{\em weak} indistinguishability security, and has ciphertexts whose
sizes grow sublinearly in the circuit size and linearly in the input
length. Formal definitions of functional encryption schemes are
provided in~\ref{sec:fenc0}.
\begin{theorem}[\cite{C:AJLMS19,EC:JLMS19,JLS19}] \label{thm:fenc0}
Let $\gamma \in (0, 1)$, $\epsilon>0$, and $D\in
\mathbb{N}$ be arbitrary constants. Let $\secparam$ be a security
parameter, $p$ be an efficiently samplable $\secparam$ bit prime, and
$k=k(\secparam)$ be a large enough positive polynomial in
$\secparam$. Assume (subexponential) hardness of 
\begin{itemize}
\item the $\mathsf{SXDH}$ assumption with respect to a bilinear groups of order $p$,
\item the $\mathsf{LWE}$ assumption with modulus-to-noise ratio $2^{k^{\epsilon}}$ where $k=k(\secparam)$ is the dimension of the secret,
\item the existence of $\gamma$-secure perturbation resilient
  generators $\drg \in (\arithNC^0,\deg\ 2)$ over $\Z_p$ with
  polynomial stretch.
\end{itemize}
There exists a secret-key functional encryption scheme for $\NC^0$
circuits with multilinear degree $D$ over $\Int$, having
\begin{itemize}
\item 1-key, weakly selective, (subexponential)
  $(\gamma+\negl)$-indistinguishability-security, and 
\item {\em sublinearly compact ciphertext with linear dependency on input
    length}, that is, ciphertext size is
  $|\ct|=\poly(\secparam)(\ilen + S^{1-\sigma})$, where $\ilen$ is the
  input length, $S$ the maximum size of the circuits supported, 
  $\sigma$ is some constant in $(0,1)$, and $\poly$ depends on
  $D$.
\end{itemize}
\end{theorem}
For convenient reference, the construction is recalled in Section
\ref{sec:fenc0}. 

\paragraph{From weak FE for $\NC^0$ to Full-Fledged FE for All Polynomial Size Circuits}
Starting from the above weak version of secret key functional
encryption scheme -- weak function class $\NC^0$, weak security,
and weak compactness -- we apply known transformations to
obtain a full-fledged {\em public key} FE scheme for polynomial size
circuits, satisfying adaptive collusion resistant security, and having
full compactness.

\begin{theorem}[Strengthening FE] \label{thm:fep}
  Let $\gamma \in (0,1)$. Let $\secparam \in \mathbb{N}$ be a security parameter and $k(\secparam)$ be a large enough positive polynomial. Assume the (subexponential) hardness of 
  \begin{itemize}
  \item the $\mathsf{LWE}$ assumption with modulus-to-noise ratio
    $2^{k^{\epsilon}}$ where $k=k(\secparam)$ is the dimension of the
    secret, and
  \item the existence of Boolean $\prg$s in $\NC^0$ with polynomial
    stretch and multilinear degree $d \in \Nat$ over $\Int$.
\end{itemize}
There are the following transformations:
\begin{enumerate}
\item \textsc{Starting Point.}

  Suppose there is a secret-key functional encryption scheme for
  $\NC^0$ circuits with multilinear {\color{red} degree $(3\pdeg + 2)$} over $\Int$,
  having 1-key, weakly selective, (subexponential)
  $\gamma$-indistinguishability security, and { sublinearly compact
    ciphertext and linear dependency on input length}.
  
\item \textsc{Lifting Function Class \cite{AJS15,FOCS:LinVai16,EC:Lin16}.}

  There exists a secret-key functional encryption scheme for {\color{red}
    polynomial size circuits}, having 1-key, weakly selective,
  (subexponential) $(\gamma+ \negl)$-indistinguishability security,
  and {\color{red} sublinearly compact ciphertexts}, that is,
  $|\ct|=\poly(\secparam, \ilen) S^{1-\sigma}$.

\item \textsc{Security Amplification \cite{AJS18,C:AJLMS19,C:JKMS20}.}

  There exists a secret-key functional encryption scheme for
  polynomial-size circuits, having 1-key, weakly selective,
  (subexponentially) {\color{red} ($\negl$-)indistinguishability security}, and
  sublinearly compact ciphertexts.

\item \textsc{Secret Key to Public Key, and Sublinear Ciphertext to Sublinear Encryption Time \cite{BNPW16,LPST16,GKPVZ13}.}

  There exists a  {\color{red} public-key} functional encryption
    scheme for polynomial size circuits, having 1-key, weakly selective, (subexponentially)
    indistinguishability security, and  {\color{red} sublinear
      encryption time}, 
    $T_\enc = \poly(\secparam, \ilen) S^{1-\sigma}$.

  \item \textsc{1-Key to Collusion Resistance \cite{TCC:GarSri16,TCC:LiMic16,EC:KitNisTan18} }

    There exists a public-key functional encryption scheme for
    polynomial-size circuits, having {\color{red} collusion resistant,
      adaptive}, (subexponentially) indistinguishability security, and
    encryption time $\poly(\secparam,\ilen)$.
  \end{enumerate}
\end{theorem}

\paragraph{FE to IO Transformation}
Finally, we rely on the FE to IO transformation to obtain $\iO$. 
\begin{theorem}[\cite{AJ15,BV15}]
\label{thm:feio}
Assume the existence of a public-key functional encryption scheme for
polynomial-size circuits, having 1-key, weakly selective,
subexponentially indistinguishability security, and sublinear
encryption time.  Then, (subexponentially secure) $i\mathcal{O}$ for
polynomial size circuits exists.
\end{theorem}

\paragraph{Putting Pieces Together} Combining Theorem \ref{thm:main},
Theorem \ref{thm:prgdrg}, Theorem \ref{thm:fenc0}, Theorem \ref{thm:fep},
and Theorem \ref{thm:feio}, we get our main result:

\begin{theorem}
\label{thm:iomain}
Let $\pstretch>1$, $\epsilon, \spasity \in (0,1)$, and $d \in \Nat$ be arbitrary
constants. Let $\secparam \in \mathbb{N}$ be a security parameter, $p$ be an
efficiently samplable $\secparam$ bit prime, and
$n=n(\secparam)$ and $k=k(\secparam)$ be large enough
positive polynomials in the security parameter.  Assume sub-exponential
hardness of the following assumptions:  
\begin{itemize}
\item the $\mathsf{LWE}$ assumption with modulus-to-noise ratio $2^{k^{\epsilon}}$ where $k$ is the dimension of the secret,
\item  the $\mathsf{SXDH}$ assumption with respect to bilinear groups of prime order $p$,
\item the existence of a Boolean $\mathsf{PRG}$ in $\mathsf{NC}^0$
  with polynomial stretch and multilinear degree $d$ over $\Int$, and
\item the $\flpn(\ell,n,\ell^{-\delta},p)$ where
  $\ell=n^{\frac{1}{\lceil \frac{d}{2}\rceil}}$.
\end{itemize}
Then, (subexponentially secure) indistinguishability obfuscation for
all polynomial-size circuits exists. Further, assuming only polynomial
security of these assumptions, there exists collusion resistant,
adaptive, and compact public-key functional encryption for all
circuits.
\end{theorem}

\subsection{Perturbation Resilient Generators}\label{sec:drg}

We recall the definition of perturbation resilient generators from~\cite{C:AJLMS19,EC:JLMS19,JLS19}.
\newcommand{\Bd}{B}
\begin{definition}[Syntax of  Perturbation Resilient Generators ($\drg$) ~\cite{C:AJLMS19,EC:JLMS19,JLS19}]
\label{def:drg}
Let $\pstretch$ be a positive constant.  A perturbation resilient
generator $\drg$ with stretch $\pstretch$ 
is defined by the following PPT algorithms: 
\begin{itemize}
\item $\setuppoly(1^{\secparam},1^n,1^\Bd):$ takes as input the security
  parameter $\secparam$, a seed length parameter $n$, and a bound
  $\Bd$, samples a function index $\pid$.
\item $\setupseed(\pid):$ samples two binary strings, a public seed
  and a private seed, $\sd = (\psd,\ssd)$. The combined length of
  these strings is $n\cdot \poly(\secparam,\log \Bd)$. 
\item $\peval(\pid, \sd):$ takes as input the index $\pid$ and the seed $\sd$ and  computes a string in $\Z^m \cap [-\poly(n,\Bd,\secparam),\ \poly(n,\Bd,\secparam)]^m$ for some fixed polynomial $\poly$.
\end{itemize}
\end{definition}

\begin{remark}
Similar to an $\sprg$, we say that $\drg$ has polynomial stretch if above $\pstretch>1$ for some constant $\pstretch$.
\end{remark}
\begin{remark}
Note that in the definition proposed by \cite{EC:JLMS19,JLS19}, the $\setupseed$ algorithm was not given as input $\pid$, however, their results still hold even if $\setupseed$ is given $\pid$ as input.
\end{remark}

\begin{definition} [Security of $\drg$ ~\cite{C:AJLMS19,EC:JLMS19,JLS19}]\label{def:drg-sec}
A perturbation resilient generator $\drg$ satisfies
\begin{description}
\item[$(T,\seclevel)$-perturbation resilience:] For every
  $n = n(\secparam)$ a positive non-constant polynomial in the
  security parameter $\secparam$, and $\Bd=\Bd(\secparam, n)$ a
  positive non-constant polynomial in $\secparam$ and $n$, and every
  sequence $\smallset{\vec \beta= \vec \beta_\secparam}$, where
  $\vec \beta \in \Z^m \cap [-\Bd,\Bd]^{m}$, we require that the following
  two distributions are
  $(T(\secparam),\seclevel(\secparam))$-indistinguishable:
\begin{align*}
\{ (\pid,\ \psd,\ \peval(\pid,\sd,\Bd)) \ | & \ \pid \gets \setuppoly(1^{\secparam},1^n,1^\Bd), \ \sd=(\ssd,\psd)\leftarrow \setupseed(\pid) \}    \\
\{ (\pid,\ \psd,\ \peval(\pid,\sd,\Bd)+\vec \beta) \ | & \ \pid \gets \setuppoly(1^{\secparam},1^n,1^\Bd),\ \sd=(\ssd,\psd)\leftarrow \setupseed(\pid) \}    
\end{align*}
\end{description}
\end{definition}

\begin{definition}[Complexity and degree of $\drg$]
  Let $\pdeg \in \Nat$, let $\secparam \in \mathbb{N}$ and
  $n=n(\secparam)$ be arbitrary positive non-constant polynomial in
  $\secparam$, and $p = p(\secparam)$ denote a prime modulus which is
  an efficiently computable function in $\secparam$.  Let $\mathbb{C}$
  be a complexity class.  A $\drg$ has complexity $\mathbb{C}$ in the
  public seed and degree $\pdeg$ in private seed over $\Int_p$,
  denoted as, $\drg \in (\mathbb C, \text{ deg } \pdeg)$, if for any
  polynomial $\Bd(n,\secparam)$ and every $\pid$ in the support of
  $\setuppoly(1^{\secparam},1^n,1^\Bd)$, there exists an algorithm
  $\pprocess_\pid$ in $\mathbb{C}$ and an $m(n)$-tuple of polynomials
  $Q_\pid$ that can be efficiently generated from $\pid$, such that
  for all $\sd$ in the support of $\setupseed(\pid)$, it holds that:
  \begin{align*}
    \peval(\pid, \sd) = Q_\pid(\ol\psd, \ssd)\; \text{over } \Int_p\; , \; \ol\psd
    = \pprocess_\pid(\psd)\; ,
  \end{align*}
  where $Q_\pid$ has degree 1 in $\ol\psd$ and degree $\pdeg$ in $\ssd$. 
  \end{definition}

We now prove the following proposition, which immediately implies Theorem~\ref{thm:prgdrg}. 
\begin{proposition}
\label{propo:prgdrg}
Assume the existence of a $(T,\gamma)$-pseudorandom structured seed
PRG, $\sprg$, in $(\mathbb{C}, \text{deg } d)$ with a stretch of $\pstretch>0$. Then for any constant  $0<\pstretch'<\pstretch$,  there exists a $(T,2\cdot \gamma+O(\frac{1}{\secparam}))$-perturbation resilient generator, $\drg$ in $(\mathbb{C},\deg \ d)$ with a stretch $\pstretch'$.
\end{proposition}
\begin{proof}
Let $\sprg$ be the given
structured-seed PRG with stretch $\pstretch$. The construction of
$\drg$ is as follows.

\begin{itemize}
\item $\drg.\setuppoly(1^{\secparam},1^n,1^\Bd):$ Run $\sprg.\pidsamp(1^{\secparam},1^n)\rightarrow \pid'$, and output $\pid=(\pid',\Bd,\secparam,n)$. 
\item $\drg.\setupseed(\pid):$ Run $\sprg.\psdsamp(\pid')\rightarrow  (\psd,\ssd)$ and output $\sd = (\psd,\ssd)$.
\item $\drg.\peval(\pid, \sd):$ Compute $\vec{z}\leftarrow \sprg.\peval(\pid',\sd)$ where
    $\vec{z}\in \{0,1\}^{n^{\pstretch}}$.  Let $m'=n^{\pstretch'}$ and
    $t=\lceil \log_2(\secparam\cdot n^{\pstretch'}\cdot \Bd)
    \rceil$. 
  \begin{itemize}
  \item If $m < m't$, there are not enough bits in the output of
    $\sprg$. Set $\vec y = \vec 0^{1\times m'}$ 
  \item Otherwise, for every $i\in [m']$, set
    $y_i=\sum_{j\in [t]} 2^{j-1}\cdot z_{(i-1)\cdot t +j} $.
  \end{itemize}
  Output $\vec y$.
\end{itemize}
\paragraph{Stretch:} The output length is exactly
$m' = n^{\pstretch'}$, while the seed length is identical to that of
$\sprg$, namely $n\poly(\secparam)$, as desired.

Further, observe that the output of $\drg$ is set to 0 when there are
not enough bits in the output of $\sprg$, namely $m < m't$. It is easy
to see that for arbitrary non-constant positive polynomials
$n=n(\secparam)$ and $\Bd=\Bd(\secparam, n)$, it holds that
$t = O(\log \secparam)$ and hence for any $0< \pstretch'< \pstretch$,
$m = n^\pstretch \ge m't =n^{\pstretch'} t$ for sufficiently large
$\secparam$. In this case, the output of $\drg$ is formed by the
output of $\sprg$.

\paragraph{Complexity:} We note that $\drg$ is in
$(\mathbb{C},\deg \ d)$. In the case that $m \ge m't$,
$\drg.\peval(\pid, \sd)$ outputs $\vec{y}$ where
$y_i=\sum_{j\in [t]} 2^{j-1}\cdot z_{(i-1)\cdot t +j} $, and
$\vec{z}=\sprg.\eval(\pid',\sd)$. Since each $y_i$ is a linear
function of $\vec{z}$ and each $z_i$ is degree $d$ in $\ssd$,
$\vec{y}$ is also degree $d$ in $\ssd$. Further since each $z_i$
is linear in $\ol\psd= \pprocess_\pid(\psd)$ and
$\pprocess_\pid \in \mathbb{C}$,  $\vec{y}$ is also linear
in $\ol\psd= \pprocess_\pid(\psd)$. In the other case that $m < m't$,
the output $\vec y = \vec 0^{1\times m'}$ and had degree 0 in both
$\psd$ and $\ssd$.  Overall, $\drg\in (\mathbb{C},\deg \ d)$.

\renewcommand{\hyb}{\mathsf{H}}

\paragraph{$(T,2 \cdot \gamma+O(\frac{1}{\secparam}))$-perturbation
  resilience:}
Fix a sufficiently large $\secparam \in \Nat$, positive non-constant
polynomials $n = n(\secparam)$, $\Bd(\secparam,n)$and
$\beta= \beta_\secparam \in \Z^m \cap [-\Bd,\Bd]^{m}$, and
$t=\log_2 (\secparam\cdot n^{\pstretch'}\cdot \Bd)$.  We now show
the perturbation resilience of $\drg$ through a sequence of hybrids.

\begin{description}
\item [Hybrid $\hyb_0$:] In this hybrid, we give to the
  adversary,
  \begin{align*}
    \forall i \in [m'], \ y_i=\sum_{j\in [t]} 2^{j-1}\cdot
    z_{(i-1)\cdot t +j}+\beta_i\; , \qquad \vec{z} = \sprg.\eval(\pid',\sd)\; , 
  \end{align*}
  along with the public index $\pid$ and the public part of the seed
  $\psd$. As observed above, when $n$ and $B$ are positive
  non-constant polynomials, and $\secparam$ is sufficiently large, it
  always holds that $m \ge m't$ and the output of $\drg$ is non-zero
  and formed as above. Thus, this hybrid corresponds to the first challenge distribution
  in the security definition of $\drg$ (Definition~\ref{def:drg-sec}).

\item [Hybrid $\hyb_1$:] In this hybrid, we change $\vec y$ to 
\begin{align*}
 y_i=\sum_{j\in [t]} 2^{j-1}\cdot r_{(i-1)\cdot
  t +j}+\beta_i~\; ,\qquad \vec{r}\gets \{0,1\}^{n^{\pstretch}}\; .
\end{align*}
This hybrid is $(T,\gamma)$-indistinguishable to hybrid $\hyb_0$ by the
$(T,\gamma)$-pseudorandomness of $\sprg$.

\item[Hybrid $\hyb_2$:] In this hybrid, we change $\vec y$ to
\begin{align*}
  y_i=u_i+\beta_i\; , \qquad u_i\gets [0,2^{t}-1]~.
\end{align*}
This hybrid  is identical to hybrid $\hyb_1$.

\item[Hybrid $\hyb_3$:] In this hybrid, we change $\vec y$ to
\begin{align*}
y_i=u_i\; , \qquad  u_i\gets [0,2^{t}-1]\; .
\end{align*}
This hybrid is statistically close to hybrid $\hyb_2$ with the
statistical distance  bounded by $O(m'\cdot \frac{\Bd}{2^{t}-1})=O(\frac{1}{n})$. This is because each
$u_i$ is uniform between $[0,2^{t}-1]$ and
$\vert \beta_i \vert\leq \Bd$.

\item[Hybrid $\hyb_4$:] In this hybrid, we change $\vec y$ to
\begin{align*}
y_i=\sum_{j\in [t]} 2^{j-1}\cdot r_{(i-1)\cdot t +j}\; , \qquad
  \vec{r}\gets \{0,1\}^{n^{\pstretch}}\; .
\end{align*}
The hybrid above is identical to hybrid $\hyb_3$.

\item[Hybrid $\hyb_5$:] In this hybrid, we give to the adversary,
\begin{align*}
y_i=\sum_{j\in [t]} 2^{j-1}\cdot z_{(i-1)\cdot t +j}\; , \qquad \vec{z} = \sprg.\eval(\pid',\sd)\; . 
\end{align*}
This hybrid is $(T,\gamma)$-indistinguishable to hybrid $\hyb_4$ by
the $(T,\gamma)$-pseudorandomness of $\sprg$. By the same argument as
in hybrid $\hyb_0$, we have $m \ge m't$ and the output of
$\drg$ is non-zero and exactly as above. Thus, this corresponds to the
second challenge distribution in Definition~\ref{def:drg-sec}.
\end{description}
By a hybrid argument, we get that the total advantage in
distinguishing the two challenge distributions in the security
definition of $\drg$ is bounded by
$2\cdot \gamma+O(\frac{1}{\secparam})$.  This concludes the proof.
\end{proof}


\section{Acknowledgements}
 
 We would like to thank Stefano Tessaro and James Bartusek for helpful discussions. 
 We would also like to thank the Simons Institute for the Theory of Computing, for hosting all three authors during the program entitled ``Lattices: Algorithms, Complexity, and Cryptography".

 Aayush Jain was partially supported by grants listed under Amit Sahai, a Google PhD fellowship and a DIMACS award. This work was partly carried out while the author was an intern at NTT Research. This work was partly carried out during a research visit conducted with support from DIMACS in association with its Special Focus on Cryptography.

Huijia Lin was supported by NSF grants CNS{-}1528178, CNS{-}1929901, CNS{-}1936825~(CAREER), the Defense
Advanced Research Projects Agency (DARPA) and Army Research Office
(ARO) under Contract No.\ W911NF-15-C-0236, and a subcontract No.\
2017-002 through Galois.

Amit Sahai was supported in part from DARPA SAFEWARE and SIEVE awards, NTT Research, NSF Frontier Award 1413955, and NSF grant 1619348, BSF grant 2012378, a Xerox Faculty Research Award, a Google Faculty Research Award, an equipment grant from Intel, and an Okawa Foundation Research Grant. This material is based upon work supported by the Defense Advanced Research Projects Agency through Award HR00112020024 and the ARL under Contract W911NF-15-C- 0205. Amit Sahai is also grateful for the contributions of the LADWP to this effort.

The views expressed are those of the authors and do not reflect the official policy or position of the Department of Defense, DARPA, ARO, Simons, Intel, Okawa Foundation, ODNI, IARPA, DIMACS, BSF, Xerox,  the National Science Foundation, NTT Research, Google, or the U.S. Government.

\pagebreak

{\small
\bibliographystyle{alpha}
\bibliography{Bibliography/abbrev3, Bibliography/crypto, Bibliography/custom, Bibliography/cryptonizkamp, Bibliography/mr, Bibliography/mathreview, Bibliography/zblatt, Bibliography/jlms}}
\appendix

\renewcommand{\set}[1]{\left\{#1\right\}}

\newcommand{\pct}{P_{\ct}}
\newcommand{\sct}{S_{\ct}}

\newcommand{\renc}{\mathsf{REnc}}
\newcommand{\reval}{\mathsf{REval}}
\newcommand{\noise}{\chi}
\newcommand{\lwesec}{\epsilon}
\renewcommand{\lwe}{k}
\section{Partially Hiding Functional Encryption}
\label{sec:fedef}
We recall the notion of Partially-hiding Functional Encryption (PHFE)
schemes; some of the text in this section is taken verbatim
from~\cite{GJLS20}.  PHFE involves functional secret keys, each of
which is associated with some 2-ary function $f$, and decryption of a
ciphertext encrypting $(\vx,
\vy)$ 
with such a key reveals $f(\vx, \vy)$, $\vx$, $f$, and nothing more
about $\vy$. Since only the input $\vy$ is hidden, such an FE scheme
is called partially-hiding FE. FE can be viewed as a special case of
PHFE where the public input is the empty string. The notion was originally
introduced by~\cite{C:GorVaiWee12} and a similar notion of
partially-hiding predicate encryption was proposed and constructed
by~\cite{C:GorVaiWee15}.

We denote functionality by $\cF:\cX\times \cY \rightarrow \cZ$. The
functionality ensemble $\cF$ as well as the message ensembles $\cX$
and $\cY$ are indexed by two parameters: $n$ and $\secparam$ (for
example $\cF_{n,\secparam}$), where $\secparam$ is the security
parameter and $n$ is a length parameter and can be viewed as a
function of $\secparam$. 
\begin{definition} (Syntax of a PHFE/FE Scheme.)
\label{def:phfesyntax}
A {\em secret key} partially hiding functional encryption scheme, $\phfe$, for the
functionality $\cF:\cX\times \cY \rightarrow \cZ$ consists of the
following polynomial time algorithms:

\begin{itemize}
\item $\ppgen(1^{\secparam},1^n):$ The public parameter generation
  algorithm is a randomized algorithm that takes as input $n$ and
  $\secparam$ and outputs a string
  $\crs$. 
\item $\Setup(\crs)$: The setup algorithm is a randomized algorithm
  that on input $\crs$, returns 
  a master secret key $\msk$.
\item
  $\Enc(\msk, (x,y) \in \cX_{n,\secparam}\times \cY_{n,\secparam})$:
  The encryption algorithm is a randomized algorithm that takes in a
  master secret key and a message $(x,y)$ and returns the ciphertext
  $\ct$ along with the input $x$. $x$ is referred to as the public
  input whereas $y$ is called the private input.
\item $\KeyGen(\msk,f \in \cF_{n,\secparam})$: The key generation
  algorithm is a randomized algorithms that takes in a description of
  a function $f \in \cF_{n,\secparam}$ and returns $\sk_f$, a
  decryption key for $f$.
\item $\Dec(\sk_f,(x,\ct))$: The decryption algorithm is a
  deterministic algorithm that returns a value $z$ in $\cZ$, or $\bot$
  if it fails.
\end{itemize}

A functional encryption scheme is a partially hiding
functional encryption scheme, where $\cX_{n,\secparam} = \emptyset$
for all $n, \secparam$. 

Define three levels of efficiency: let $S=S(\secparam, n)$ be
the maximum size of functions in $\cF_{\secparam, n}$; ciphertext
$\ct$ produced by running $\ppgen, \Setup, \Enc$ honestly as above has
the following sizes with respect to some arbitrary constant
$\epsilon \in (0,1]$.
\begin{itemize}
  \item { Sublinear compactness:} $
  \poly(\secparam, n) S^{1-\epsilon}  $
 \item { Sublinear compactness and linear dependency on input
  length:} $ \poly(\secparam)(n+ S^{1-\epsilon})$
 \item { Linear Efficiency:} $ \poly(\secparam)n$
\end{itemize}

\end{definition}
We surpress the public input in notation in the case of functional encryption.



\begin{definition} (Correctness of a PHFE/FE scheme.)
\label{def:phfecorrectness}
A secret key partially hiding functional encryption scheme, $\phfe$, for the
functionality $\cF:\cX\times \cY \rightarrow \cZ$ is correct if for
every $\secparam\in \mathbb{N}$ and every polynomial
$n(\secparam) \in \mathbb{N}$, for every
$(x,y)\in \cX_{n,\secparam} \times \cY_{n,\secparam}$ and every
$f\in \cF_{n,\secparam}$, we have:
\begin{align*}
\Pr \left [  \dec(\sk_f,x,\ct))= f(x,y)\ \Bigg|\ \begin{array}{l}
           \ppgen(1^{\secparam},1^n)\rightarrow \crs\\
            \setup(\crs)\rightarrow \msk\\
            \enc(\msk,(x,y))\rightarrow (x,\ct)\\
            \keygen(\msk,f)\rightarrow \sk_f\\
\end{array}\right ] = 1 
\end{align*}
\end{definition}

\begin{definition}[Simulation security]\label{def:sim}
  \label{def:simsecurity}
  A secret-key partially hiding functional encryption scheme $\phfe$ for functionality
  $\cF:\cX \times \cY \rightarrow \cZ$ is (weakly selective) $(T,\epsilon)$-SIM secure, if for
  every positive polynomials $n = n(\secparam)$,
  $Q_\ct = Q_\ct(\secparam)$, $Q_\sk=Q_\sk(\secparam)$,
  ensembles $\smallset{(x,y)}$, $\smallset{\smallset{(x_{i},y_{i})}_{i \in [Q_\ct]}}$ in
  $\cX_{\secparam, n}\times \cY_{\secparam, n}$ and $\smallset{\smallset{f_j}_{j \in [Q_\sk]}}$ in
  $\cF_{\secparam, n}$, the following distributions are $(T,
  \epsilon)$-indistinguishable. 
  \begin{align*} 
    & \left\{ \left(\crs,\ \ct ,\ \smallset{\ct_i}_{i \in[Q_\ct]},\ \smallset{\sk_j}_{j \in [Q_\sk]}\right)\ \Bigg |
    \begin{array}{l}
      \crs \leftarrow \ppgen(1^\secparam,1^n),\
      \msk \gets \Setup(\crs)\\
      \ct \gets \Enc(\msk,(x,y))\\
      \forall i \in [Q_\ct],\ \ct_i \gets \Enc(\msk,(x_{i},y_i))\\
      \forall j \in [Q_{\sk}],\ \sk_j \gets \KeyGen(\msk,f_j)
    \end{array}\right\}\\
    & \left\{ \left(\crs,\ \tilde \ct,\ \smallset{\tilde \ct_i}_{i \in[Q_\ct]},\
      \smallset{\tilde \sk_j}_{j \in [Q_\sk]}\right)\ \Bigg |
    \begin{array}{l}
      \crs \leftarrow \ppgen(1^\secparam,1^n),\
      \tilde \msk \gets \tilde \Setup(\crs)\\
      \tilde \ct \gets \tilde\Enc_1(\tilde \msk,{\color{red} x})\\
      \forall i \in [Q_\ct],\ \tilde\ct_i \gets \tilde\Enc_2(\tilde \msk,(x_{i},y_i))\\
      \forall j \in [Q_{\sk}],\ \tilde \sk_j \gets \KeyGen(\tilde \msk,f_j,{\color{red} f_j(x,y)})
    \end{array}\right\}
  \end{align*}
\end{definition}

\begin{definition}[Indistinguishability security]\label{def:ind}
  \label{def:indsecurity}
  A secret-key functional encryption scheme $\fe$ for functionality
  $\cF:\cX \rightarrow \cZ$ is (weakly selective) $(T,\epsilon)$-IND secure, if for
  every positive polynomials $n = n(\secparam)$,
  $Q_\ct = Q_\ct(\secparam)$, $Q_\sk=Q_\sk(\secparam)$,
  ensembles $\smallset{\smallset{x_{i,0},x_{i,0}}_{i \in [Q_\ct]}}$ in
  $\cX_{\secparam, n}$ and $\smallset{\smallset{f_j}_{j \in [Q_\sk]}}$ in
  $\cF_{\secparam, n}$, the following distributions for $b \in \zo$ are $(T,
  \epsilon)$-indistinguishable. 
  \begin{align*} 
    \left\{ \left(\crs,\ \smallset{\ct_i}_{i \in[Q_\ct]},\ \smallset{\sk_j}_{j \in [Q_\sk]}\right)\ \Bigg |
    \begin{array}{l}
      \crs \leftarrow \ppgen(1^\secparam,1^n),\
      \msk \gets \Setup(\crs)\\         
      \forall i \in [Q_\ct],\ \ct_i \gets \Enc(\msk,{\color{red} x_{i,b}})\\
      \forall j \in [Q_{\sk}],\ \sk_j \gets \KeyGen(\msk,f_j)
    \end{array}\right\}
  \end{align*}
\end{definition}

\section{Recap of constant-depth functional encryption}
\label{sec:fenc0}
We give a self-contained description of a construction of {\em 1-key}
secret-key $\fe$ for $\NC^0$ satisfying {\em sublinear compactness
  with linear dependency on input length}, which can be transformed to $\iO$ as described in Section~\ref{sec:iotheorems}. We
emphasize that the construction of $\fe$ for $\NC^0$ recalled here was
given by prior works~\cite{C:AJLMS19,EC:JLMS19,FOCS:LinVai16,EC:Lin16}.  The purpose of
this appendix is providing a clean and self-contained description of
the construction for convenient lookup, and we omit the security
proof.

Consider the class of $\NC^0$ functions
$g: \zo^\ilen \rightarrow \zo^\olen$. Such functions can be computed
by a multilinear polynomial with 1/-1 coefficient of some constant
degree $D$.  We now describe the $\fe$ scheme for computing such
functions, which uses the following ingredients.

\paragraph{Ingredients.} Let $\secparam$ be the security parameter and
$p = p(\secparam) = O(2^\secparam)$ an efficiently computable prime
modulus.
\begin{itemize}
\item $\mathsf{LWE}$ over $\Int_p$ with subexponential modulus to
  noise ratio $2^{\lwe^\lwesec}$ where $\lwe$ is the dimension of LWE secret
  and $\lwesec$ is some arbitrary constant in $(0,1)$. 
  
  \noindent {\em Related parameters are set to:} 
  \begin{itemize}
  \item We use polynomially large noises: Let $\noise_{\alpha, \Bd}$
    be the truncated discrete gaussian distribution with parameter
    $\alpha$ and support $[-\Bd,\Bd]\cap \Z$, where $\alpha \le \Bd$
    are set appropriately and of magnitude $\poly(\secparam)$. As such,
    the modulus-to-noise ratio is $p/\poly(\secparam)$. 
  \item Set the LWE dimension $\lwe$ appropriately
    $\lwe = \Theta(\secparam^{1/\lwesec})$ such that the
    modulus-to-noise ratio $p/\poly(\secparam)$ is upper bounded by
    $2^{\lwe^{\lwesec}}$.
  \end{itemize}

  We will use the basic homomorphic encryption scheme by~\cite{FOCS:BraVai11}
  based on LWE. An encryption of a Boolean string $\vec
  x$ has form $\vec A, \vec b = \vec s \vec A + 2\vec e +  \vec x$
  over $\Int_p$ and supports homomorphic evaluation of constant degree
  polynomials over $\Int_p$ (without relinearization). 
  
\item A perturbation resilient generator
  $\drg = (\setuppoly,\setupseed,\peval)$ with stretch $\tau > 1$ and
  complexity $(\arithNC^{1}, \deg \ 2)$ over $\Z_p$. Such a $\drg$ was
  constructed in Section \ref{sec:iotheorems}, based on  Boolean $\prg$s in $\NC^0$ the LPN assumption over $\Int_p$.

  \noindent {\em Related parameters are set to:} 
  \begin{itemize}
  \item The bound on the noises to be smudged is set to be $\Bd^{D}\cdot l^{D}\cdot \secparam$.

  \item The output length of $\drg$ is $m$, matching the output length
    of the $\NC^0$ computation. 
  \item The seed length is then $n\poly(\secparam)$ for $n =
    m^{1/\tau}$. 
  \end{itemize}
  
\item A SIM-secure  collusion-resistant secret-key
   scheme for
   $(\arithNC^1, \text{ deg } 2)$, $\phfe=(\phfe.\ppgen, \phfe.\setup, \phfe.\enc,\phfe.\keygen, \phfe.\dec )$. This can be built
  from the $\mathsf{SXDH}$ assumption over asymmetric bilinear groups
  of order $p$ as presented in \cite{JLS19}.

  \noindent {\em Related parameters are set to:} 
  \begin{itemize}
  \item The input length parameter $n'$ is an efficiently computable
    function depending on $n, \lwe, D$ set implicitly in the
     $\Enc$ algorithm below. 
  \end{itemize} 
\end{itemize}

\paragraph{Construction:} The $\NC^0$-FE scheme $\fe = (\ppgen,
\setup, \enc, \keygen, \dec)$ is as follows: 
\begin{description}
\item [$\crs \leftarrow \ppgen(1^{\secparam},1^\ilen)$:] Sample  $\vec A\leftarrow \Z^{\lwe \times \ilen}_p$, 
     $\crs_{\phfe} \leftarrow
      \phfe.\ppgen(1^{\secparam},1^{n'})$, \\ and 
     $\pid \leftarrow \drg.\setuppoly(1^{\secparam}, 1^n, 1^{\Bd^{\fedeg}\cdot l^{\fedeg}\cdot \secparam})$.
    Output $\crs=(\crs_{\phfe},\pid, \vec A)$.
    
\item [$\msk \leftarrow \setup(\crs)$:] Sample $\msk_\phfe\leftarrow
  \phfe.\setup(\crs_{\phfe})$ and output $\msk = (\msk_\phfe, \crs)$.

\item [$\ct \leftarrow \enc(\msk,\vec x \in \{0,1\}^\ilen)$:]$ $ 
  \begin{itemize}

  \item Sample $(\psd,\ssd) \gets\drg.\setupseed(\pid)$. Note that the
    seed has length $|\psd|+|\ssd| = n \poly(\secparam)$.

  \item Encrypt $\vec x$ as follows: Sample a secret
    $\vec s \gets \Z^{\lwe}_p$ and noise vector
    $\vec e \gets \noise^\ilen_{\alpha,B}$, and compute
    $\vec b= \vec s \vec A + 2\vec e + \vec x$. 
    
  \item Let $\overline{{\vec s}}=(1 \Vert \vec s)$ and compute 
    $\overline{\vec s}^{\otimes \lceil\frac{D}{2}\rceil}$. 
    
  \item Set public input $X = (\psd, \vec b)$ and private input
    $Y=(\ssd,\overline{\vec s}^{\otimes \lceil\frac{D}{2} \rceil})$,
    and encrypt them using $\phfe$,
    $\ct \gets \phfe.\enc(\msk, (X,Y))$.
  \end{itemize}
  Output $\ct$.

\item [$\sk \gets \keygen(\msk,g)$:] Output a $\phfe$ key
  $\sk_\phfe \gets \phfe.\keygen(\msk,G)$ for the following function
  $G$.

  \underline{\bf Function $G$} takes public input $X$ and private input $Y$ and does the following:
  \begin{itemize}
  \item Compute $f(\vec x) + 2\vec e'$ via a polynomial $G^{(1)}$ that
    has degree $D$ in $X$ and degree 2 in $Y$.

    \underline{\bf Function $G^{(1)}$} is defined as follows:
    Since $f$ is a degree $D$ multilinear polynomial with 1/-1
    coefficients, we have (using the same notation as in
    Section~\ref{sec:sprg}) 
    \begin{align*}
      \forall j \in [m], \ f_j(\vec x )= L_j((x_\mnlv)_{\mnlv \in
      f_j})\ \text{ for some linear $L_j$ with 1/-1 coefficients}   ~.
    \end{align*}     
    The decryption equation for $\vec b$ is 
    \begin{align*}
      \forall i \in [\ilen],\ x_i + 2e_i & = \innerp{ \vec c_i,\ \ol{\vec s}} & \vec c_i = -\vec  a^{\transpose}_i || b_i,
                                                 \ \vec a_i \text{ is the $i$th column of  $\vec A$}~.
    \end{align*}           
    Thus, we have
    \begin{align*}
    \forall \text{ degree $D$ monomial } \mnlv ,\ x_\mnlv + 2e_\mnlv & = \innerp{ \otimes_{i \in \mnlv} \vec c_i,\ \otimes_{i
                                                      \in \mnlv} \ol{\vec  s}} & \\
      \forall j \in [m], \ f_j(\vec x ) + 2e'_j & = L_j\left(\left(\innerp{ \otimes_{i \in \mnlv} \vec c_i,\ \otimes_{i
        \in \mnlv} \ol{\vec  s}}\right)_{\mnlv\in f_j}\right) &\\
        e'_j & = L_j((e_\mnlv)_{\mnlv \in f_j})\text{ has
                       $\poly(\secparam)$ magnitude} &
    \end{align*}
    Define $G^{(1)}$ to be the polynomial that computes
    $f(\vec x) + 2\vec e'$
    \begin{align*}
       G^{(1)}(X, Y) = f(\vec x) + 2\vec e'~,
    \end{align*}
    with degree $D$ in $X$ (containing $\vec b$) and degree 2 in $Y$
    (containing $\overline{\vec s}^{\otimes
      \lceil\frac{D}{2}\rceil}$). $G^{(1)}$ also depends on $\vec A$.
           
  \item Compute $\vec r \gets \drg.\peval(\pid,\sd)$. 
  \item Output $\vec y' = \vec y + 2\vec e_f + 2\vec r $.
  \end{itemize}
  Observe that because of the complexity of $G^{(1)}$ and $\drg$,  $G$ is in $(\arithNC^1, \text{deg } 2)$.
  
\item [$\dec(\sk,\ct)$:] Decrypt the PHFE ciphertext
  $\vec y + 2\vec e' = G(X, Y) \gets \phfe.\dec(\sk_\phfe,
  \ct_\phfe)$, which reveals $\vec y \ \mathrm{mod}\ 2$.

  More precisely, the decryption of PHFE built from bilinear groups
  produces $g_T^{(y_j + 2e'_j)}$ for every $j \in[m]$, where $g_T$ is the
  generator of the target group. Thus, decryption needs to first
  extracts $y_j + 2e'_j$ by brute force discrete logarithm, which is
  efficient as $e'_j$ has $\poly(\secparam)$ magnitude.
\end{description}

\paragraph{Sublinear Compactness with Linear Dependency on Input
  Length} Observe that the ciphertext $\ct$ produced above has size
$\poly(\secparam, \ilen) S^{1-\epsilon} = \poly(\secparam, \ilen)
m^{1-\epsilon}$ for some $\epsilon \in (0,1)$, following from the following facts:
\begin{itemize}
\item By the linear
efficiency of PHFE,  $|\ct| = \poly(\secparam)(|X|+|Y|)$.
\item The seed $P, S$ of $\drg$ has length $m^{1/\tau}$ for $\tau >1$.
\item $|\vec b| = k \log p = O(k \secparam)$.
\item $\overline{\vec s}^{\otimes \lceil\frac{D}{2}\rceil}$ has size
  $\lwe^{\lceil\frac{D}{2}\rceil} \log p = O(\secparam^{
    (\lceil\frac{D}{2}\rceil/\lwesec) + 1})= \poly(\secparam)$.
\end{itemize}





\end{document}